\documentclass[12pt,a4paper]{amsart}
\usepackage{etex,xr}
\usepackage[margin=1in]{geometry}
\usepackage{amsthm,amsmath,amssymb,amsfonts,mathtools}
\usepackage[shortcuts]{extdash}
\usepackage{graphicx}

\let\counterwithin\relax
\usepackage{chngcntr,apptools}
\usepackage{comment}
\usepackage[round,compress]{natbib}
\usepackage{mathrsfs}

\let\oldFootnote\footnote
\newcommand{\nextToken}{\relax}
\renewcommand{\footnote}[1]{\oldFootnote{#1}\futurelet\nextToken\isFootnote}
\newcommand{\isFootnote}{\ifx\footnote\nextToken\textsuperscript{,}\fi}

\usepackage{color}
\usepackage{setspace}
\usepackage[unicode=true,pdftitle={The Equilibrium Existence Duality:
Equilibrium with Indivisibilities \& Income Effects},pdfauthor={Elizabeth Baldwin, Omer Edhan, Ravi Jagadeesan, Paul Klemperer, and Alexander Teytelboym}, bookmarks=true,bookmarksnumbered=false,bookmarksopen=false, breaklinks=false,pdfborder={0 0 1},colorlinks=false]{hyperref}
\usepackage{tikz-cd}
\usepackage[all]{xy}

\newcommand\citeposs[1]{\citeauthor{#1}'s (\citeyear{#1})}

\newtheorem{theorem}{Theorem}
\newtheorem*{theorem*}{Theorem}
\newtheorem{claim}{Claim}
\newtheorem{corollary}{Corollary}
\newtheorem{lemma}{Lemma}
\newtheorem{proposition}{Proposition}
\newtheorem{fact}{Fact}

\newtheorem{innercustomgeneric}{\customgenericname}
\providecommand{\customgenericname}{}
\newcommand{\newcustomtheorem}[2]{%
  \newenvironment{#1}[1]
  {\renewcommand\customgenericname{#2}%
   \renewcommand\theinnercustomgeneric{\ref*{##1}$'$}%
   \innercustomgeneric
  }
  {\endinnercustomgeneric}
}
\newcustomtheorem{primetheorem}{Theorem}

\theoremstyle{definition}
\newtheorem{definition}{Definition}

\newtheorem{innercustomgenericdef}{\customgenericname}
\providecommand{\customgenericname}{}
\newcommand{\newcustomdef}[2]{%
  \newenvironment{#1}[1]
  {\renewcommand\customgenericname{#2}%
   \renewcommand\theinnercustomgenericdef{\ref*{##1}$'$}%
   \innercustomgenericdef
  }
  {\endinnercustomgenericdef}
}

\newcustomdef{primedefinition}{Definition}

\newtheorem{remark}{Remark}
\newtheorem{example}{Example}

\AtAppendix{\counterwithin{fact}{section}}
\AtAppendix{\counterwithin{lemma}{section}}
\AtAppendix{\counterwithin{proposition}{section}}
\AtAppendix{\counterwithin{corollary}{section}}
\AtAppendix{\counterwithin{definition}{section}}
\AtAppendix{\counterwithin{theorem}{section}}
\AtAppendix{\counterwithin{claim}{section}}
\AtAppendix{\counterwithin{remark}{section}}
\AtAppendix{\counterwithin{example}{section}}
\AtAppendix{\counterwithin{equation}{section}}

\newcommand\coloneq{=}

\usepackage{caption}
\usepackage[labelformat=simple]{subcaption}

\newlength{\tempparskip}
\newlength{\halfparskip}
\usepackage{calc,enumitem}
\newenvironment{casework}{
    \begin{enumerate}[labelindent=15pt,labelwidth=\widthof{Case 9: },leftmargin=\tempparskip, itemindent=\widthof{Case 9: },listparindent=\tempparskip]
    \addtolength{\itemindent}{\halfparskip}
    }
{\end{enumerate}}

\DeclareFontFamily{OT1}{pzc}{}
\DeclareFontShape{OT1}{pzc}{m}{it}{<-> s * [1.10] pzcmi7t}{}
\DeclareMathAlphabet{\mathpzc}{OT1}{pzc}{m}{it}

\renewcommand\emptyset\varnothing
\DeclareMathOperator*\argmax{arg\,max}
\DeclareMathOperator*\argmin{arg\,min}

\newcommand\ssm\smallsetminus
\newcommand\toto{\rightrightarrows}

\newcommand{\zerodel}{.\kern-\nulldelimiterspace}
\newcommand\lgiv{\,\left|\,}
\newcommand\lgivend{\right\zerodel}
\newcommand\rgiv{\,\right|\,}
\newcommand\rgivend{\left\zerodel}

\newcommand\linset{L}

\DeclareMathOperator\conv{Conv}
\newcommand\e[1]{{\bf e}^{#1}}
\newcommand{\Z}{\mathbb{Z}}
\newcommand{\R}{\mathbb{R}}

\newcommand{\zero}{\mathbf{0}}
\renewcommand\v{{\bf v}_I}
\newcommand\dvec{\mathbf{d}}

\newcommand\norm{\mathbf{g}}
\newcommand\normComp[1]{g_{#1}}

\newcommand\numersub{0}

\newcommand\p{\mathbf{p}_{I}}
\newcommand\pall{\mathbf{p}}
\newcommand\pnumer{p_\numersub}
\newcommand\ppr{\p'}
\newcommand\pprall{\pall'}
\newcommand\pComp[1]{p_{#1}}
\newcommand\pprComp[1]{\pComp{#1}'}
\newcommand\hp{\hat{\mathbf{p}}_{I}}
\newcommand\hpComp[1]{\hat{p}_{#1}}

\newcommand\bunj{\bunag{j}}

\newcommand\bunag[1]{\bun^{#1}}

\newcommand\bun{\mathbf{x}_{I}}

\newcommand\bunComp[1]{x_{#1}}
\newcommand\hbun{\hat{\mathbf{x}}_I}

\newcommand\hbunj{\hbunag{j}}

\newcommand\hbunag[1]{\hbun^{#1}}

\newcommand\bunpr{\bun'}
\newcommand\bunprComp[1]{\bunComp{#1}'}
\newcommand\bundpr{\bun''}

\newcommand\tot{\mathbf{y}_{I}}
\newcommand\totComp[1]{y_{#1}}
\newcommand\bundowj{\bundowag{j}}

\newcommand\bundowag[1]{\bundow^{#1}}

\newcommand\bundow{\mathbf{w}_{I}}
\newcommand\bundowComp[1]{w_{#1}}

\newcommand\bunnj{\bunnag{j}}
\newcommand\bunnag[1]{\bunn^{#1}}
\newcommand\numerj{\numerag{j}}
\newcommand\numerag[1]{\numer^{#1}}
\newcommand\bunn{\mathbf{x}}
\newcommand\bunnpr{\bunn'}
\newcommand\bunndpr{\bunn''}
\newcommand\numer{x_{\numersub}}
\newcommand\numerpr{\numer'}
\newcommand\numerdpr{\numer''}
\newcommand\numertpr{\numer'''}
\newcommand\defbunn{\numer,\bun}
\newcommand\hbunnj{\hbunnag{j}}
\newcommand\hbunnag[1]{\hbunn^{#1}}
\newcommand\hnumerj{\hnumerag{j}}
\newcommand\hnumerag[1]{\hnumer^{#1}}
\newcommand\hbunn{\hat{\mathbf{x}}}
\newcommand\hnumer{\hat{x}_{\numersub}}
\newcommand\bunndowj{\bunndowag{j}}
\newcommand\bunndowag[1]{\bunndow^{#1}}
\newcommand\numerdowj{\numerdowag{j}}
\newcommand\numerdowag[1]{w^{#1}_{\numersub}}
\newcommand\bunndow{\mathbf{w}}
\newcommand\numerdow{w_{\numersub}}
\newcommand\defbunndow{\numerdow,\bundow}

\newcommand\cfFn[1]{S^{#1}}
\newcommand\cf[3]{S^{#1}\left(#2;#3\right)}
\newcommand\dHFn[1]{D^{#1}_{\mathrm{H}}}
\newcommand\dH[3]{\dHFn{#1}\left(#2;#3\right)}
\newcommand\dMFn[1]{D_{\mathrm{M}}^{#1}}
\newcommand\dM[3]{D_{\mathrm{M}}^{#1}\left(#2,#3\right)}
\newcommand\dQLFn[1]{D^{#1}}
\newcommand\dQL[2]{\dQLFn{#1}\left(#2\right)}
\newcommand\Feas[1]{X^{#1}_I}
\newcommand\feas[1]{\underline{x}^{#1}_0}
\newcommand\Feans[1]{X^{#1}}
\newcommand\feasUtil[1]{(\minu{#1},\maxu{#1})}

\newcommand\ub{u}
\newcommand\ubj{\ubag{j}}
\newcommand\ubag[1]{u^{#1}}
\newcommand\hubj{\hubag{j}}
\newcommand\hubag[1]{\hat{u}^{#1}}
\newcommand\ubvec{{\bf \ub}}
\newcommand\hubvec{{\bf \hubag{}}}
\newcommand\util[2]{\utilFn{#1}\left(#2\right)}
\newcommand\utilFn[1]{U^{#1}}
\newcommand\val[2]{\valFn{#1}\left(#2\right)}
\newcommand\valFn[1]{V^{#1}}
\newcommand\quasival[2]{\quasivalFn{#1}\left(#2\right)}
\newcommand\quasivalFn[1]{V_{\mathrm{Q}}^{#1}}
\newcommand\valH[3]{\valHFn{#1}\left(#2; #3\right)}
\newcommand\valHFn[1]{V^{#1}_{\mathrm{H}}}
\newcommand\valHDef[1]{\valH{#1}{\cdot}{\ub}}
\newcommand\maxu[1]{\overline{u}^{#1}}
\newcommand\minu[1]{\underline{u}^{#1}}
\newcommand\umax[1]{u_{\mathrm{max}}^{#1}}
\newcommand\umin[1]{u_{\mathrm{min}}^{#1}}
\newcommand\payoff[1]{t^{#1}}
\newcommand\payoffvec{{\bf \payoff{}}}
\newcommand\payoffs[1]{\payoffFn(#1)}
\newcommand\payoffFn{T}

\newcommand\ce{competitive equilibrium}

\newcommand\CE{Competitive Equilibrium}
\newcommand\ces{competitive equilibria}
\newcommand\Ces{Competitive equilibria}

\newcommand\dtvs{demand type vector set}
\newcommand\dualecon{Hicksian economy}
\newcommand\dualecons{Hicksian economies}

\newcommand\DualEcons{Hicksian Economies}
\newcommand\dualval{Hicksian valuation}
\newcommand\dualvals{Hicksian valuations}
\newcommand\DualVal{Hicksian Valuation}

\newcommand\gsubst{gross \subst}

\newcommand\gsubstity{gross \substity}
\newcommand\Gsubstity{Gross \substity}
\newcommand\GSubstity{Gross \Substity}
\newcommand\gcompity{gross \compity}

\newcommand\nsubst{net \subst}

\newcommand\nsubstity{net \substity}

\newcommand\Nsubstity{Net \substity}
\newcommand\NSubstity{Net \Substity}

\newcommand\osubstity{ordinary \substity}
\newcommand\Osubstity{Ordinary \substity}

\newcommand\ssubst{strong \subst}
\newcommand\ssubstity{strong \substity}

\newcommand\SSubstity{Strong \Substity}
\newcommand\snsubst{strong \nsubst}

\newcommand\snsubstity{strong \nsubstity}

\newcommand\subst{substitutes}

\newcommand\substity{substitutability}
\newcommand\Substity{Substitutability}
\newcommand\compity{complementarity}

\newcommand\totpriceeffs{price effects}
\newcommand\intvec{integer vector}
\newcommand\intvecs{integer vectors}

\newcommand\dowalloc{endowment allocation}
\newcommand\andowalloc{an \dowalloc}
\newcommand\Andowallocemph{An \emph{\dowalloc}}
\newcommand\dowallocs{\dowalloc s}

\newcommand\myspacing

\title[The Equilibrium Existence Duality]{The Equilibrium Existence Duality:\\
Equilibrium with Indivisibilities \& Income Effects}\thanks{An abstract of this paper will appear in the \emph{Proceedings of the 21st ACM Conference on Economics and Computation (EC'20)}.  We thank Federico Echenique, Scott Kominers, Michihiro Kandori, Th\`anh Nguyen,  	
Wolfgang Pesendorfer, Rakesh Vohra, and many seminar participants 
for their valuable comments on this paper.
}
\author[Baldwin, Edhan, Jagadeesan, Klemperer, and Teytelboym]{Elizabeth Baldwin \and Omer Edhan \and Ravi Jagadeesan \and Paul Klemperer \and Alexander Teytelboym}
\date{17th June 2020}
\thanks{Baldwin: Department of Economics and Hertford College, University of Oxford; {\tt elizabeth.baldwin@economics.ox.ac.uk}.
Edhan: Department of Economics, University of Manchester; {\tt omer.edhan@gmail.com}.
Jagadeesan: Harvard Business School; and Department of Economics, Harvard University; {\tt ravi.jagadeesan@gmail.com}.
Klemperer: Department of Economics and Nuffield College, University of Oxford; {\tt paul.klemperer@nuffield.ox.ac.uk}.
Teytelboym: Department of Economics, Institute for New Economic Thinking, and St.~Catherine's College, University of Oxford; {\tt alexander.teytelboym@economics.ox.ac.uk}.
Jagadeesan was supported by a National Science Foundation Graduate Research Fellowship under grant number DGE-1745303, and by the Washington Center for Equitable Growth.
Teytelboym was supported by the Economic and Social Research Council grant number ES/R007470/1.
Parts of this work were done while Jagadeesan and Teytelboym were visiting the Simons Institute for the Theory of Computing, and while Jagadeesan was visiting Nuffield College, Oxford.}

\begin{document}

\begin{abstract}
We show that, with indivisible goods, the existence of competitive equilibrium fundamentally depends on agents' substitution effects, not their income effects. Our Equilibrium Existence Duality allows us to transport results on the existence of competitive equilibrium from settings with transferable utility to settings with income effects. One consequence is that net substitutability---which is a strictly weaker condition than gross substitutability---is sufficient for the existence of competitive equilibrium.  We also extend the ``demand types'' classification of valuations to settings with income effects and give necessary and sufficient conditions for a pattern of substitution effects to guarantee the existence of competitive equilibrium.

\vspace{12pt}

\end{abstract}
\maketitle

\clearpage
\myspacing

\section{Introduction}

This paper shows that, when goods are indivisible and there are income effects, the existence of competitive equilibrium fundamentally depends on agents' substitution effects---i.e., the effects of compensated price changes on agents' demands.  We provide general existence results that do not depend on income effects.

In contrast to the case of divisible goods, competitive equilibrium does not generally exist in settings with indivisible goods \citep{henry1970indivisibilites}.
Moreover, most previous results about when equilibrium does exist with indivisible goods assume that utility is transferable---ruling out income effects but allowing tractable characterizations of (Pareto-)efficient allocations and aggregate demand that can be exploited to analyze competitive equilibrium.\footnote{For example, methods based on integer programming (see, e.g., \cite{koopmans1957assignment}, \cite{BiMa:97}, \cite{Ma:98}, \cite{CaOzPa:15}, and \cite{tran2019product}) rely on characterizations of the set of Pareto-efficient allocations as the solutions to a welfare maximization problem, while methods based on convex programming (see, e.g., \cite{Muro:2003}, \cite{ikebe2015stability}, and \cite*{CaEpVo:17}) and tropical geometry \citep{BaKl:14,BaKl:19} rely on representing aggregate demand as the demand of a representative agent.}
But understanding the role of income effects is important for economies with indivisible goods, as these goods may comprise large fractions of agents' budgets.  Furthermore, in the presence of income effects, the distribution of wealth among agents affects both Pareto efficiency and aggregate demand, making it necessary to develop new methods to analyze competitive equilibrium with indivisible goods.

The cornerstone of our analysis is an application of the relationship between Marshallian and Hicksian demand.  As in classical demand theory, Hicksian demand is defined by fixing a utility level and minimizing the expenditure of obtaining it.  We combine Hicksian demands to construct a family of ``Hicksian economies'' in which prices vary, but agents' utilities---rather than their endowments---are held constant.  Our key result, which we call the Equilibrium Existence Duality, states that competitive equilibria exist for all endowment allocations if and only if competitive equilibria exist in the Hicksian economies for all utility levels.

Preferences in each Hicksian economy reflect agents' substitution effects.  Therefore, by the Equilibrium Existence Duality, the existence of competitive equilibrium fundamentally depends on substitution effects.  Moreover, as fixing a utility level precludes income effects, agents' preferences are quasilinear in each Hicksian economy.  Hence, the Equilibrium Existence Duality allows us to transport (and so generalize) \emph{any} necessary or sufficient condition for equilibrium existence from settings with transferable utility to settings with income effects.\footnote{Outside the case of substitutes (which we describe in detail), \cite{BiMa:97} and \cite{Ma:98} gave necessary and sufficient conditions  on profiles of valuations, and \cite{CaOzPa:15} gave sufficient conditions on agents' individual valuations, for the existence of competitive equilibrium in transferable utility economies.}
In particular, our most general existence result 
gives a necessary and sufficient condition for a pattern of agents' substitution effects to guarantee the existence of competitive equilibrium in the presence of income effects.

Consider, for example, the case of substitutable goods in which each agent demands at most one unit of each good.
With transferable utility, substitutability is sufficient for the existence of competitive equilibrium \citep{KeCr:82} and defines a maximal domain for existence \citep{GuSt:99}.  With income effects, \citet{FlJaJaTe:19} showed that competitive equilibrium exists under gross substitutability.
The Equilibrium Existence Duality tells us that, with income effects, competitive equilibrium in fact exists under \emph{net} substitutability and that net substitutability defines a maximal domain for existence.
Moreover, we show that gross substitutability implies net substitutability; the reverse direction is not true in the presence of income effects.

An implication of our results is that it is unfortunate that \cite{KeCr:82}, and much of the subsequent literature, used the term ``gross substitutes'' to refer to a condition on quasilinear preferences.
Indeed, gross and net substitutability are equivalent without income effects, and our work shows that it is net substitutability, not gross substitutability, that is critical to the existence of competitive equilibrium with substitutes.\footnote{\cite{KeCr:82} were aware of the equivalence between gross and net substitutability in their setting (see their Footnote 1) but used the term ``gross substitutes'' due to an analogy of their arguments for existence with t\^{a}tonnement from general equilibrium theory.}

To appreciate the distinction between gross and net substitutability, suppose that Martine owns a house and is thinking about selling her house and buying one of two different other houses: a spartan one and a luxurious one \citep{quinzii1984core}.
If the price of her own house increases, she may wish to buy the luxurious house instead of the spartan one---exposing a gross complementarity between her existing house and the spartan one.
However, Martine regards the houses as net substitutes: the complementarity emerges entirely due an income effect.
Competitive equilibrium is therefore guaranteed to exist in economies with Martine if all other agents see the goods as net substitutes, despite the presence of gross complementarities.

Our most general equilibrium existence theorem characterizes the combinations of substitution effects that guarantee the existence of competitive equilibrium.
It is based on \citeposs{BaKl:19} classification of valuations into ``demand types.''
A demand type is defined by the set of vectors that summarize the possible ways in which demand can change in response to a small generic price change.
For example, the set of all substitutes valuations forms a demand type, as does the set of all complements valuations, etc.

Applying \citeauthor{BaKl:19}'s taxonomy to changes in Hicksian demands, we see that their definition easily extends to general utility functions, capturing agents' substitution effects.
Examples of demand types in our setting with income effects, therefore, include the set of all net substitutes preferences, the set of all net complements preferences, etc.
The Equilibrium Existence Duality then makes it straightforward that the Unimodularity Theorem\footnote{See Theorem 4.3 of \cite{BaKl:19}; an earlier version was given by \cite{DaKoMu:01}.}---which encompasses many standard results on the existence of competitive equilibrium as special cases\footnote{It generalizes the quasilinear case of \cite{KeCr:82}, and results of 
\cite{SuYa:06}, \cite{MiSt:09}, \cite{HaKoNiOsWe:11}, and \cite{Teyt:14}.}---is unaffected by income effects.
Therefore, as with the case of substitutes, conditions on complementarities and substitutabilities 
that guarantee the existence of competitive equilibrium in settings with transferable utility translate to conditions on net complementarities and substitutabilities that guarantee the existence of competitive equilibrium in settings with income effects.
In particular, there are patterns of net complementarities that are compatible with the existence of competitive equilibrium.

Our results may have significant implications for the design of auctions that seek competitive equilibrium outcomes, and in which bidders face financing constraints.
For example, they suggest that versions of the Product-Mix Auction \citep{klemperer2008new}, used by the Bank of England since the Global Financial Crisis, may work well in this context.

Several other papers have considered the existence of competitive equilibrium in the presence of indivisibilities and income effects.
\cite{quinzii1984core}, \cite{gale1984equilibrium}, and \cite{svensson1984competitive} showed the existence of competitive equilibrium in a housing market economy in which agents have unit demand and endowments.
Building on those results, \cite{kaneko1986existence}, \cite{van1997existence,van2002existence}, and \cite{yang2000equilibrium} analyzed settings with multiple goods, but restricted attention to separable preferences.
By contrast, our results---even for the case of substitutes---allow for interactions between the demand for different goods. We also clarify the role of net substitutability for the existence of competitive equilibrium.  

In a different direction, \cite{DaKoMu:01} proved a version of the sufficiency direction of the Unimodularity Theorem for settings with income effects.
\cite{DaKoMu:01} also defined domains of preferences using an optimization problem that turns out to be equivalent to the expenditure minimization problem.
However, they did not note the connection to the expenditure minimization problem or Hicksian demand, and, as a result,
did not interpret their sufficient conditions in terms of substitution effects or establish the role of substitution effects in determining the existence of equilibrium.

We proceed as follows.
Section~\ref{sec:setting} describes our setting---an exchange economy with indivisible goods and money.
Section~\ref{sec:EED} develops the Equilibrium Existence Duality.
Since the existing literature has focused mostly on the case in which indivisible goods are substitutes, we consider that case in Section~\ref{sec:subst}.
Section~\ref{sec:demTypes} develops demand types for settings with income effects 
and states our Unimodularity Theorem with Income Effects.
Section~\ref{sec:auctions} remarks on implications for auction design, and Section~\ref{sec:conclusion} is a conclusion.
Appendix~\ref{app:EEDproof} proves the Equilibrium Existence Duality. Appendix~\ref{app:grossToNet} proves the connection between gross and net \substity.
Appendices~\ref{app:dualDemPrefs} and~\ref{app:maxDomain} 
adapt the proofs of results from the literature to our setting.

\section{The Setting}
\label{sec:setting}

We work with a model of exchange economies with indivisibilities---adapted to allow for income effects.
There is a finite set $J$ of agents, a finite set $I$ of indivisible goods, and a divisible num\'eraire that we call ``money.''
We allow goods to be undesirable, i.e., to be ``bads."
We fix a \emph{total endowment} $\tot \in \mathbb{Z}^I$ of goods in the economy.\footnote{In particular, we allow for multiple units of some goods to be present in the aggregate, unlike \cite{GuSt:99} and \cite{CaOzPa:15}.}

\subsection{Preferences and Marshallian Demand}

Each agent $j \in J$ has a finite set $\Feas{j} \subseteq \Z^I$ of \emph{feasible bundles} of indivisible goods and a lower bound $\feas{j} \ge -\infty$ on her consumption of money.
As bundles that specify negative consumption of some goods can be feasible, our setting implicitly allows for production.\footnote{Technological constraints on production (in the sense of \cite{HaKoNiOsWe:11} and \cite{FlJaJaTe:19}) can be represented by the possibility that some bundles of goods are infeasible for an agent to consume (see Example 2.15 in \cite{BaKl:14}).}
The principal cases of $\feas{j}$ are $\feas{j} = -\infty$, in which case all levels of consumption of money are feasible, and $\feas{j} = 0$, in which case the consumption of money must be positive.
Hence, the set of feasible consumption bundles for agent $j$ is $\Feans{j} = (\feas{j},\infty) \times \Feas{j}$.
Given a bundle $\bunn \in \Feans{j},$ we let $\numer$ denote the amount of money in $\bunn$ and $\bun$ denote the bundle of goods specified by $\bunn,$ so $\bunn = (\defbunn)$.

The utility levels of agent $j$ lie in the range $\feasUtil{j}$, where $-\infty \le \minu{j} < \maxu{j} \le \infty.$
Furthermore, each agent $j$ has a \emph{utility function} $\utilFn{j}: \Feans{j} \to \feasUtil{j}$ that we assume to be continuous and strictly increasing in $\numer$, and to satisfy
\begin{equation}
\label{eq:ulimits}
\lim_{\numer \to (\feas{j})^+} \util{j}{\defbunn} = \minu{j} \quad \text{and} \quad \lim_{\numer \to \infty} \util{j}{\defbunn} = \maxu{j}
\end{equation}
for all $\bun \in \Feas{j}.$
Condition (\ref{eq:ulimits}) requires that some consumption of money above the minimum level $\feas{j}$ be essential to agent $j$.\footnote{\citet[pages 543--544]{henry1970indivisibilites}, \citet[Theorem 1(i)]{mas1977indivisible}, and \citet[Equation (3.1)]{DeGa:85} made similar assumptions.
If consuming money is inessential but consumption of money must be nonnegative, then it is known that \ce\ may not exist \citep{mas1977indivisible}---even in settings in which agents have unit demand for goods (see, e.g., \cite{herings2019competitive}).  However, the existence of \ce\ can be guaranteed when the agents trade lotteries over goods \citep{Gul.etal2020}.}
We let $\pnumer = 1$.

Given an endowment $\bunndow = (\defbunndow) \in \Feans{j}$ of a feasible consumption bundle and a price vector $\p \in \mathbb{R}^I,$ agent $j$'s \emph{Marshallian demand} for goods is
\[\dM{j}{\p}{\bunndow} = \left\{\bun^* \lgiv \bunn^* \in \argmax_{\bunn \in X^j \mid \pall \cdot \bunn \le \pall \cdot \bunndow} \util{j}{\bunn}\lgivend\right\}.\]
As usual, Marshallian demand is given by the set of bundles of goods that maximize an agent's utility, subject to a budget constraint, given a price vector and an endowment.
An \emph{income effect} is a change in an agent's Marshallian demand induced by a change in her money endowment, holding prices fixed.\footnote{Note that income effects also correspond to changes in an agent's Marshallian demand induced by changes in the value of her endowment, holding prices fixed.}

Our setup is flexible enough to capture a wide range of preferences with and without income effects, as the following two examples illustrate.

\begin{example}[Quasilinear Utility]
\label{eg:quasilin}
Given a \emph{valuation} $\valFn{j}: \Feas{j} \to \mathbb{R},$ letting $\feas{j} = \minu{j}= -\infty$ and $\maxu{j} =\infty,$ one obtains a quasilinear utility function given by
\[\util{j}{\defbunn} = \numer + \val{j}{\bun}.\]
When agents utility functions are quasilinear, they do not experience income effects.
When all agents have quasilinear utility functions, we say that utility is \emph{transferable}.
\end{example}

\begin{example}[Quasilogarithmic Utility]
\label{eg:quasilog}
Given a function $\quasivalFn{j}: \Feas{j} \to (-\infty,0)$, which we call a \emph{quasivaluation},\footnote{Here, we call $\quasivalFn{j}$ a quasivaluation, and denote it by $\quasivalFn{j}$ instead of $\valFn{j}$, to distinguish it from the valuation of an agent with quasilinear preferences.} and letting $\minu{j}= -\infty$, $\maxu{j} =\infty,$ and $\feas{j} = 0,$ there is a \emph{quasilogarithmic} utility function given by
\[\util{j}{\bunn} = \log \numer - \log(- \quasival{j}{\bun}).\]
Unlike with quasilinear utility functions, agents with quasilogarithmic utility functions exhibit income effects.
\end{example}

\subsection{Hicksian Demand, \DualVal{s}, and the \DualEcons}

The concept of Hicksian demand from consumer theory plays a key role in our analysis.
Given a utility level $\ub \in \feasUtil{j}$ and a price vector $\p,$ agent $j$'s \emph{Hicksian demand} for goods is
\begin{equation}
\label{eq:costMin}
\dH{j}{\p}{\ub} = \left\{\bun^* \lgiv \bunn^* \in \argmin_{\bunn \in \Feans{j} \mid \util{j}{\bunn} \ge \ub} \pall \cdot \bunn\lgivend\right\}.
\end{equation}
As in the standard case with divisible goods, Hicksian demand is given by the set of bundles of goods that minimize the expenditure of obtaining a utility level given a price vector.
A \emph{substitution effect} is a change in an agent's Hicksian demand induced by a change in prices, holding her utility level fixed.

As in classical demand theory, Marshallian and Hicksian demand are related by the duality between the utility maximization and expenditure minimization problems.
Specifically, 
a bundle of goods is expenditure-minimizing if and only if it is utility-maximizing.\footnote{Although Fact~\ref{fac:dualDem} is usually stated with divisible goods (see, e.g., Proposition 3.E.1 and Equation (3.E.4) in \cite{MaWhGr:95}), the standard proof applies with multiple indivisible goods and money under Condition~(\ref{eq:ulimits}).
For sake of completeness, we give a proof of Fact~\ref{fac:dualDem} in Appendix~\ref{app:dualDemPrefs}.}

\begin{fact}[Relationship between Marshallian and Hicksian Demand]
\label{fac:dualDem}
Let $\p$ be a price vector.
\begin{enumerate}[label=(\alph*)]
\item For all endowments $\bunndow$, we have that $\dM{j}{\p}{\bunndow} = \dH{j}{\p}{\ub},$ where
\[\ub = \max_{\bunn \in \Feans{j} \mid \pall \cdot \bunn \le \pall \cdot \bunndow} \util{j}{\bunn}.\]
\item For all utility levels $\ub$ and endowments $\bunndow$ with
\[\pall \cdot \bunndow = \min_{\bunn \in \Feans{j} \mid \util{j}{\bunn} \ge \ub} \pall \cdot \bunn,\]
we have that $\dH{j}{\p}{\ub} = \dM{j}{\p}{\bunndow}.$
\end{enumerate}
\end{fact}

If an agent has a  quasilinear utility function, then, as she experiences no income effects, her Marshallian and Hicksian demands coincide and do not depend on endowments or utility levels.
Under quasilinearity, we therefore refer to both Marshallian and Hicksian demand simply as \emph{demand}, which we denote by $\dQL{j}{\p}$.
Formally, if $j$ has quasilinear utility with valuation $\valFn{j},$ defining $\dQL{j}{\p}$ as the solution to the quasilinear maximization problem
\begin{equation}\label{eqn:quasilin}
\dQL{j}{\p} = \argmax_{\bun \in \Feas{j}} \{\val{j}{\bun} - \p \cdot \bun\},
\end{equation}
we have that $\dM{j}{\p}{\bunndow} = \dQL{j}{\p}$ for all endowments $\bunndow$ and that $\dH{j}{\p}{\ub} = \dQL{j}{\p}$ for all utility levels $\ub$.

We next show that the interpretation of the expenditure minimization problem as a quasilinear maximization problem persists in the presence of income effects.
Specifically, we can rewrite the expenditure minimization problem of Equation (\ref{eq:costMin}) as a quasilinear optimization problem by using the constraint to solve for $\numer$ as a function of $\bun$.
Formally, for a bundle $\bun \in \Feas{j}$ of goods and a utility level $\ub \in \feasUtil{j},$ we let $\cf{j}{\bun}{\ub} = \util{j}{\cdot,\bun}^{-1}(\ub)$ denote the level of consumption of money (or \emph{s}avings) needed to obtain utility level $\ub$ given $\bun.$\footnote{The function $\cfFn{j}$ is the \emph{compensation function} of \cite{DeGa:85} (see also \cite{DaKoMu:01}).} By construction, we have that
\[\dH{j}{\p}{\ub} = \argmin_{\bun \in \Feas{j}} \left\{\cf{j}{\bun}{\ub} + \p \cdot \bun\right\}.\]
It follows that agent $j$'s expenditure minimization problem at utility level $\ub$ can be written as a quasilinear maximization problem for the valuation $-\cf{j}{\cdot}{\ub}$, which we therefore call the Hicksian valuation.

\begin{definition}
The \emph{\dualval} of agent $j$ at utility level $u$ is $\valHDef{j} = -\cf{j}{\cdot}{\ub}$.
\end{definition}

Note that $\cf{j}{\cdot}{\ub}$ is continuous and strictly increasing in $\ub,$ and hence $\valHDef{j}$ is continuous and strictly decreasing in $\ub$.
The following lemma formally states that agent $j$'s Hicksian demand at utility level $\ub$ is the demand correspondence of an agent with valuation $\valHDef{j}$.

\begin{lemma}
\label{lem:dHvalH}
For all price vectors $\p$ and utility levels $\ub$, we have that
\[\dH{j}{\p}{\ub} = \argmin_{\bun \in \Feas{j}} \left\{\cf{j}{\bun}{\ub} + \p \cdot \bun\right\} = \argmax_{\bun \in \Feas{j}} \left\{\valH{j}{\bun}{\ub} - \p \cdot \bun\right\}.\]
\end{lemma}
\begin{proof}
As $\utilFn{j}(\bunn)$ is strictly increasing in $\numer,$ we have that
\[\dH{j}{\p}{\ub} = \left\{\bun^* \lgiv \bunn^* \in \argmin_{\bunn \in \Feans{j} \mid \util{j}{\bunn} = \ub} \pall \cdot \bunn\lgivend\right\}.\]
Applying the substitution $\numer = \cf{j}{\bun}{\ub}=-\valH{j}{\bun}{\ub}$ to remove the constraint from the minimization problem yields the lemma.
\end{proof}

It follows from Lemma~\ref{lem:dHvalH} that an agent's \dualval\ at a utility level gives rise to a quasilinear utility function that reflects the agent's substitution effects at that utility level.
Lemma~\ref{lem:dHvalH} also yields a relationship between the family of \dualvals\ and income effects.
Indeed, by Fact~\ref{fac:dualDem}, an agent's income effects correspond to changes in her Hicksian demand induced by changes in her utility level, holding prices fixed.
By Lemma~\ref{lem:dHvalH}, these changes in Hicksian demand reflect the changes in the \dualval\ that are induced by the changes in utility levels.
Hence, the \dualvals\ at each utility level determine an agent's substitution effects, while the variation of the \dualvals\ with the utility level captures her income effects.

To illustrate how an agent's family of \dualvals\ reflects her income effects, we consider the cases of quasilinear and quasilogarithmic utility.

\begin{example}[Example~\ref{eg:quasilin} continued]
With quasilinear utility, the \dualval\ at utility level $\ub$ is $\valH{j}{\bun}{\ub} = \val{j}{\bun} - \ub.$
Changes in $\ub$ do not affect the relative values of bundles under $\valH{j}{\cdot}{\ub}$, so changes in the utility level do not affect Hicksian demand.
Indeed, there are no income effects.
By construction, a utility function $\util{j}{\bunn}$ is quasilinear in $\numer$ if and only if $\cf{j}{\bun}{\ub}$ is quasilinear in $\ub$---or, equivalently, $\valH{j}{\bun}{\ub}$ is quasilinear in $-\ub$.
\end{example}

In general, it follows from Fact~\ref{fac:dualDem} and Lemma~\ref{lem:dHvalH} that agent $j$'s preferences exhibit income effects if and only if $\cf{j}{\bun}{\ub}$---or, equivalently, $\valH{j}{\bun}{\ub}$---is not additively separable between $\bun$ and $\ub$.

\begin{example}[Example~\ref{eg:quasilog} continued]
\label{eg:quasilogDualVal}
With quasilogarithmic utility, the \dualval\ at utility level $\ub$ is
$\valH{j}{\bun}{\ub} = e^{\ub} \quasival{j}{\bun}.$
In this case, each \dualval\ is a positive linear transformation of $\quasivalFn{j}$.
Income effects are reflected by the fact that $\valH{j}{\bun}{\ub}$ is not additively separable between $\bun$ and $\ub$.
\end{example}

We use Lemma~\ref{lem:dHvalH} to convert preferences with income effects into families of valuations.
It turns out that each continuously decreasing family of valuations is the family of \dualvals\ of a utility function, so a utility function can be represented equivalently by a family of \dualvals.

\begin{fact}[Duality for Preferences]
\label{fac:dualPrefs}
Let $F: \Feas{j} \times \feasUtil{j} \to (-\infty,-\feas{j})$ be a function.
There exists a utility function $\utilFn{j}: \Feans{j} \to \feasUtil{j}$ whose \dualval\ at each utility level $\ub$ is $F(\cdot,\ub)$ if and only if for each $\bun \in \Feas{j},$ the function $F(\bun,\cdot)$ is continuous, strictly decreasing, and satisfies\footnote{A version of Fact~\ref{fac:dualPrefs} for the function $\cfFn{j}$ in a setting in which utility is increasing in goods is proved in Lemma 1 in \cite{DaKoMu:01}.  For sake of completeness, we give a proof of Fact~\ref{fac:dualPrefs} in Appendix~\ref{app:dualDemPrefs}.

Fact~\ref{fac:dualPrefs} is also similar in spirit to the duality between utility functions and expenditure functions (see, e.g., Propositions 3.E.2 and 3.H.1 in \cite{MaWhGr:95}).  However, the arguments of the expenditure function (at each utility level) are prices, while the arguments of the \dualval\ (at each utility level) are quantities.}\footnote{Condition (\ref{eq:vlimits}) is analogous to Condition (\ref{eq:ulimits}) and ensures that the corresponding utility function is defined everywhere on $\Feans{j}$.  Note that Condition (\ref{eq:vlimits}) is essentially automatic in the context of \cite{DaKoMu:01} and therefore does not appear explicitly in their result (Lemma 1 in \cite{DaKoMu:01}).}
\begin{equation}
\label{eq:vlimits}
\lim_{\ub \to (\minu{j})^+} F(\bun,\ub) = -\feas{j} \quad \text{and} \quad \lim_{\ub \to (\maxu{j})^-} F(\bun,\ub) = -\infty.
\end{equation}
\end{fact}

Finally, we combine the families of \dualval{s} to form a family of \dualecons, in each of which utility is transferable and agents choose consumption bundles to minimize the expenditure of obtaining given utility levels.

\begin{definition}
The \emph{\dualecon\ for a profile of utility levels $(\ubj)_{j \in J}$} is the transferable utility economy in which agent $j$'s valuation is $\valH{j}{\cdot}{\ubj}$.
\end{definition}

The family of \dualecons\ consists of the ``duals" of the original economy in which income effects have been removed and \totpriceeffs\ are given by substitution effects.
Like the construction of \dualvals, the construction of the \dualecons\ allows us to convert economies with income effects to families of economies with transferable utility and is a key step of our analysis.

\section{The Equilibrium Existence Duality}
\label{sec:EED}

We now turn to the analysis of \ce\ in exchange economies. 
\Andowallocemph\ consists of an endowment $\bunndowj \in \Feans{j}$ for each agent $j$ such that $\sum_{j \in J} \bundowj = \tot,$ where $\tot$ is the total endowment.
Given \andowalloc, a \ce\ specifies a price vector such that markets for goods clear when agents maximize utility. 
By Walras's Law, it follows that the market for money clears as well.

\begin{definition}
Given \andowalloc\ $(\bunndowj)_{j \in J}$, a \emph{\ce} consists of a price vector $\p$ and a bundle $\bunj \in \dM{j}{\p}{\bunndowj}$ for each agent such that $\sum_{j \in J} \bunj = \tot$.
\end{definition}

In transferable utility economies, a \ce\ consists of a price vector $\p$ and a bundle $\bunj \in \dQL{j}{\p}$ for each agent such that $\sum_{j \in J} \bunj = \tot$.
In this case, the \dowalloc\ does not affect \ce\ because endowments do not affect (Marshallian) demand.
We therefore omit the \dowalloc\ when considering \ce\ in transferable utility economies in which \andowalloc\ exists---i.e., $\tot \in \sum_{j \in J} \Feas{j}$.
On the other hand, the total endowment $\tot$ affects \ce\ even when utility is transferable.

Recall that utility is transferable in the \dualecons.
Furthermore, by Lemma~\ref{lem:dHvalH}, a \ce\ in the \dualecon\ for a profile $(\ubj)_{j \in J}$ of utility levels consists of a price vector $\p$ and a bundle $\bunj \in \dH{j}{\p}{\ubj}$ for each agent such that $\sum_{j \in J} \bunj = \tot$.
Thus, agents act as if they minimize expenditure in \ce\ in the \dualecons.\footnote{\label{fn:quasiEquil}As a result, \ces\ in the \dualecons\ coincide with \emph{quasiequilibria with transfers} from the modern treatment of the Second Fundamental Theorem of Welfare Economics (see, e.g., Definition 16.D.1 in \cite{MaWhGr:95}).
As the set of feasible levels of money consumption is open, agents always can always reduce their money consumption slightly from a feasible bundle to obtain a strictly cheaper feasible bundle.
Hence, quasiequilibria with transfers coincide with equilibria with transfers in the original economy (see, e.g., Proposition 16.D.2 in \cite{MaWhGr:95} for the case of divisible goods).
If the endowments of money were fixed in the \dualecons, this concept would coincide with the concept of \emph{compensated equilibrium} of \cite{arrow1971general} and the concept of \emph{quasiequilibrium} introduced by \cite{debreu1962new}.
}

Building on Fact~\ref{fac:dualDem} and Lemma~\ref{lem:dHvalH},
our Equilibrium Existence Duality connects the equilibrium existence problems in the original economy (which can feature income effects) and the \dualecon\ (in which utility is transferable).
Specifically, we show that \ce\ always exists in the original economy if and only if it always exists in the \dualecons.
Here, we hold agents' preferences and the total endowment (of goods) fixed but allow the \dowalloc\ to vary.


\begin{theorem}[Equilibrium Existence Duality]
\label{thm:existDualExchange}
Suppose that the total endowment and the sets of feasible bundles are such that \andowalloc\ exists.
\Ces\ exist for all \dowallocs\ if and only if \ces\ exist in the \dualecons\ for all profiles of utility levels.
\end{theorem}

By Lemma~\ref{lem:dHvalH}, agents' substitution effects determine their preferences in each \dualecon.  Therefore, Theorem~\ref{thm:existDualExchange} tells us that \emph{any} condition that ensures the existence of \ces\ can be written as a condition on substitution effects alone.
That is, substitution effects fundamentally determine whether \ce\ exists.

Both directions of Theorem~\ref{thm:existDualExchange} also have novel implications for the analysis of \ce\ in economies with indivisibilities.
As demands in the \dualecons\ are given by Hicksian demand in the original economy (Lemma~\ref{lem:dHvalH}), 
the ``if" direction of Theorem~\ref{thm:existDualExchange} implies that every condition on demand $\dQLFn{j}$ that guarantees the existence of \ce\ in settings with transferable utility translates into a condition on \emph{Hicksian} demand $\dHFn{j}$ that guarantees the existence of \ce\ in settings with income effects.
In Sections~\ref{sec:subst} and~\ref{sec:demTypes}, we use the ``if" direction of Theorem~\ref{thm:existDualExchange} to obtain new domains for the existence of \ce\ with income effects from previous results on the existence of \ce\ in settings with transferable utility \citep{KeCr:82,BaKl:19}.
Conversely, the ``only if" direction of Theorem~\ref{thm:existDualExchange} shows that if a condition on demand defines a maximal domain for the existence of \ce\ in settings with transferable utility, then the translated condition on Hicksian demand defines a maximal domain for the existence of \ce\ in settings with income effects.
In Sections~\ref{sec:subst} and~\ref{sec:demTypes}, we also use this implication to derive new maximal domain results for settings with income effects.

To prove the ``only if" direction of Theorem~\ref{thm:existDualExchange}, we exploit a version of the Second Fundamental Theorem of Welfare Economics for settings with indivisibilities.
To understand connection to the existence problem for the \dualecons, note that the existence of \ce\ in the \dualecons\ is equivalent to the conclusion of the Second Welfare Theorem---i.e., that each Pareto-efficient allocation can be supported in an equilibrium with endowment transfers---as the following lemma shows.\footnote{Recall that an allocation $(\bunnj)_{j \in J} \in \bigtimes_{j \in J} \Feans{j}$ is \emph{Pareto-efficient} if there does not exist an allocation $(\hbunnj)_{j \in J} \in \bigtimes_{j \in J} \Feans{j}$ such that
\[\sum_{j \in J} \hbunnj = \sum_{j \in J} \bunnj,\]
and $\util{j}{\hbunnj} \ge \util{j}{\bunnj}$ for all agents $j$ with strict inequality for some agent.}

\begin{lemma}
\label{lem:dualEconSWT}
Suppose that the total endowment and the sets of feasible bundles are such that \andowalloc\ exists.
\Ces\ exist in the \dualecons\ for all profiles of utility levels if and only if, for each Pareto-efficient allocation $(\bunnj)_{j \in J}$ with $\sum_{j \in J} \bunj = \tot$,
there exists a price vector $\p$ such that $\bunnj \in \dM{j}{\p}{\bunnj}$ for all agents $j$.
\end{lemma}

We prove Lemma~\ref{lem:dualEconSWT} in Appendix~\ref{app:EEDproof}.
Intuitively, as utility is transferable in the \dualecons, variation in utility levels between \dualecons\ plays that same role as endowment transfers in the Second Welfare Theorem. 
It is well-known that the conclusion of the Second Welfare Theorem holds whenever \ces\ exist for all \dowallocs\ \citep{maskin2008fundamental}.\footnote{While \cite{maskin2008fundamental} assumed that goods are divisible, their arguments apply even in the presence of indivisibilities---as we show in Appendix~\ref{app:EEDproof}.}
It follows that \ce\ always exists in the \dualecons\ whenever it always exists in the original economy, which is the ``only if" direction of Theorem~\ref{thm:existDualExchange}.

We use a different argument to prove the ``if" direction.
Our strategy is to show that there exists a profile of utility levels and a \ce\ in the corresponding \dualecon\ in which all agents' expenditures equal their budgets in the original economy.
To do so, we apply a topological fixed-point argument that is similar in spirit to standard proofs of the existence of \ce.
Specifically, we consider an auctioneer who, for a given profile of candidate equilibrium utility levels, evaluates agents' expenditures over all \ces\ in the \dualecon\ and adjusts candidate equilibrium utility levels upwards (resp.~downwards) for agents who under- (resp.~over-) spend their budgets.\footnote{This approach is similar in spirit to \citeposs{negishi1960welfare} proof of the existence of \ce\ with divisible goods.
\cite{negishi1960welfare} instead applied an adjustment process to the inverses of agents' marginal utilities of money.
However, \citeposs{negishi1960welfare} approach does not generally yield a convex-valued adjustment process in the presence of indivisibilities.}
The existence of \ce\ in the \dualecons\ ensures that the process is nonempty-valued, and the transferability of utility in the \dualecons\ ensures that the process is convex-valued.
Kakutani's Fixed Point Theorem implies the existence of a fixed-point utility profile.
By construction, there exists a \ce\ in the corresponding \dualecon\ at which agents' expenditures equal the values of their endowments.
By Lemma~\ref{lem:dHvalH}, agents must be maximizing utility given their endowments at this equilibrium, and hence once obtains a \ce\ in the original economy.
The details of the argument are in Appendix~\ref{app:EEDproof}.

\subsection{Examples}

We next illustrate the power of Theorem~\ref{thm:existDualExchange} using the two examples.

Our first example is a ``housing market'' in which agents have unit-demand preferences, may be endowed with a house, and can experience arbitrary income effects.
We can use Theorem~\ref{thm:existDualExchange} to reduce the existence problem to the assignment game of \cite{koopmans1957assignment}---reproving a result originally due to \cite{quinzii1984core}.

\begin{example}[A Housing Market---\citealp{quinzii1984core,gale1984equilibrium,svensson1984competitive}]
\label{eg:house}
For each agent $j,$ let $\Feas{j} \subseteq \{\zero\} \cup \{\e{i} \mid i \in I\}$ be nonempty.
In this case, in \dualecon, utility is transferable and agents have unit demand for the goods.
As the \dowalloc\ does not affect \ce\ when utility is transferable,
the results of \cite{koopmans1957assignment} imply that \ces\ exist in the \dualecons\ for all profiles of utility levels (provided that \andowalloc\ exists).
Hence, Theorem~\ref{thm:existDualExchange} implies that \ces\ exist for all endowment allocations---even in the presence of income effects.
\end{example}

In the second example, we revisit the quasilogarithmic utility functions from Example~\ref{eg:quasilog}.
We provide sufficient conditions on agents' quasivaluations for \ce\ to exist.
These conditions are related to, but not in general implied by, the conditions developed in Sections~\ref{sec:subst} and~\ref{sec:demTypes}.

\begin{example}[Existence of \CE\ with Quasilogarithmic Preferences]
\label{eg:quasilogExist}
For each agent $j,$ let $\quasivalFn{j}: \Feas{j} \to (-\infty,0)$ be a quasivaluation.
Let agent $j$'s utility function be quasilogarithmic for the quasivaluation $\quasivalFn{j}$, as in Example~\ref{eg:quasilog}. In this case, agent $j$'s \dualval\ at each utility level is a positive linear transformation of $\quasivalFn{j}$ (Example~\ref{eg:quasilogDualVal}).
Hence, by Theorem~\ref{thm:existDualExchange}, \ces\ exist for all \dowallocs\ as long as \ce\ exists when utility is transferable and each agent $j$'s valuation is an (agent-dependent) positive linear transformation of $\quasivalFn{j}$---e.g., if the quasivaluations $\quasivalFn{j}$ are all \ssubst\ valuations \citep{MiSt:09}, or all valuations of a unimodular demand type \citep{BaKl:19}.
Additionally, in the case in which one unit of each good is available in total (i.e., $\totComp{i} = 1$ for all goods $i$), \cite{CaOzPa:15} showed that \ce\ exists when utility is transferable and all agents have sign-consistent tree valuations.
Hence, if one unit of each good is available in total, then Theorem~\ref{thm:existDualExchange} implies that \ces\ exist with quasilogarithmic utility for all \dowallocs\ if all agents' quasivaluations are sign-consistent tree valuations.
\end{example}

In the remainder of the paper, we use Theorem~\ref{thm:existDualExchange} to develop novel conditions on preferences that ensure the existence of \ce.

\section{The Case of Substitutes}
\label{sec:subst}

In this section, we apply the Equilibrium Existence Duality (Theorem~\ref{thm:existDualExchange}) to prove a new result regarding the existence of \ce\ with substitutable indivisible goods and income effects: we show that a form of \emph{net} \substity\ is sufficient for, and in fact defines a maximal domain for, the existence of \ce.
We begin by reviewing previous results on the existence of \ce\ under (gross) \substity.
We then derive our existence theorem for \nsubstity\ and relate it to the previous results. 

In this section, we focus on the case in which each agent demands at most one unit of each good.
Formally, we say that an agent $j$ \emph{demands at most one unit of each good} if $\Feas{j} \subseteq \{0,1\}^I$.
We extend to the case in which agents can demand multiple units of some goods in Section~\ref{sec:ssub}.

\subsection{\GSubstity\ and the Existence of \CE}
\label{sec:subOld}

We recall a notion of \gsubstity\ for preferences over indivisible goods from \cite{FlJaJaTe:19}, which extends
the \gsubstity\ condition from classical demand theory.
It requires that \emph{uncompensated} increases in the price of a good weakly raise demand for all other goods.
With quasilinear utility, the modifier ``gross'' can be dropped---as in classical demand theory (see also Footnote 1 in \cite{KeCr:82}).

\begin{definition}[\GSubstity]
\label{def:gsub}
Suppose that agent $j$ demands at most one unit of each good.
\begin{enumerate}[label=(\alph*)]
\item \label{part:gsub}
A utility function $\utilFn{j}$ is a \emph{\gsubst\ utility function at endowment $\bundow \in \Feas{j}$ of goods} if for all money endowments $\numerdow > \feas{j}$, price vectors $\p$, and $\lambda > 0,$ whenever $\dM{j}{\p}{\bunndow} = \{\bun\}$ and $\dM{j}{\p + \lambda \e{i}}{\bunndow} = \{\bunpr\},$ we have that $\bunprComp{k} \ge \bunComp{k}$ for all goods $k \not= i$.\footnote{Our definition of \gsubstity\ holds the endowment of goods fixed, but, unlike \cite{FlJaJaTe:19}, imposes a condition at every feasible endowment of money.  Imposing the ``full substitutability in demand language'' condition from Assumption D.1 in Supplemental Appendix D of \cite{FlJaJaTe:19} at every money endowment is equivalent to our \gsubstity\ condition.}
\item A \emph{\subst\ valuation} is a valuation for which the corresponding quasilinear utility function is a \gsubst\ utility function.
\footnote{Note that \substity\ is independent of the endowment of goods as endowments do not affect the demands of agents with quasilinear utility functions.  Our definition of \substity\ coincides with \citeposs{KeCr:82} definition \citep*{DaKoLa:2003}.}
\end{enumerate}
\end{definition}

Technically, Definition~\ref{def:gsub} imposes a \substity\ condition on the locus of prices at which Marshallian demand is single-valued---following \cite{AuMi:02}, \cite{HaKoNiOsWe:11}, \cite{BaKl:19}, and \cite{FlJaJaTe:19}.\footnote{By contrast, \cite{KeCr:82} imposed a \gsubstity\ condition at all price vectors.
Imposing \citeposs{KeCr:82} condition at every money endowment leads to a strictly stronger condition than Definition~\ref{def:gsub}\ref{part:gsub} in the presence of income effects \citep{schlegel2018trading}.}

It is well-known that when utility is transferable, \ce\ exists under \substity.

\begin{fact}
\label{fac:subExist}
Suppose that utility is transferable and that \andowalloc\ exists.
If each agent demands at most one unit of each good and has a \subst\ valuation, then \ce\ exists.\footnote{Fact~\ref{fac:subExist} is a version of Theorem 1 in \cite{HaKoNiOsWe:11} for exchange economies and follows from Proposition 4.6 in \cite{BaKl:19}.  See \cite{KeCr:82} and \cite{GuSt:99} for earlier versions that assume that valuations are monotone.}
\end{fact}

Moreover, the class of \subst\ valuations forms a maximal domain for the existence of \ce\ in transferable utility economies.
Specifically, if an agent has a non-\subst\ valuation, then \ce\ may not exist when the other agents have \subst\ valuations.
Technically, we require that one unit of each good be present among agents' endowments (i.e., that $\totComp{i} = 1$ for all goods $i$) as complementarities between goods that are not present are irrelevant for the existence of \ce.

\begin{fact}
\label{fac:subMaxDomain}
Suppose that $\totComp{i} = 1$ for all goods $i$.
If $|J| \ge 2$, agent $j$ demands at most one unit of each good, and $\valFn{j}$ is not a \subst\ valuation, then there exist sets $\Feas{k} \subseteq \{0,1\}^I$ of feasible bundles and \subst\ valuations $\valFn{k}: \Feas{k} \to \mathbb{R}$ for agents $k \not= j$, for which there exists \andowalloc\ but no \ce.\footnote{Fact~\ref{fac:subMaxDomain} is a version of Theorem 2 in \cite{GuSt:99} and Theorem 4 in \cite{yang2017maximal} that applies when $\Feas{k}$ can be strictly contained in $\{0,1\}^I$, as well as a version of Theorem 7 in \cite{HaKoNiOsWe:11} for exchange economies.
For sake of completeness, we give a proof of Fact~\ref{fac:subMaxDomain} in Appendix~\ref{app:maxDomain}.
The proof shows that the statement would hold if $|J| \geq |I|$ and agents $k \not= i$ were restricted to unit-demand valuations---as in Theorem 2 in \cite{GuSt:99}.}
\end{fact}

While Fact~\ref{fac:subMaxDomain} shows that there is no domain strictly containing the domain of \subst\ valuations for which the existence of \ce\ can be guaranteed in transferable utility economies, it does not rule out the existence of other domains for which the existence of \ce\ can be guaranteed.
For example, \cite{SuYa:06}, \cite{CaOzPa:15}, and \cite{BaKl:19} gave examples of domains other than \substity\ for which the existence of \ce\ is guaranteed.

Generalizing Fact~\ref{fac:subExist} to settings with income effects, \cite{FlJaJaTe:19} showed that \ce\ exists for \andowalloc\ $(\bunndowj)_{j \in J}$ if each agent $j$'s utility function is a \gsubst\ utility function at her endowment $\bundowj$ of goods.\footnote{\cite{FlJaJaTe:19} worked with a matching model and considered equilibrium with personalized pricing, but their arguments also apply in exchange economies without personalized pricing.
However, \cite{FlJaJaTe:19} only required that each agent sees goods as \gsubst\ for a fixed endowment of goods and money.
Our notion of \gsubstity\ considers a fixed endowment of goods but a variable endowment of money, and therefore the existence result of \cite{FlJaJaTe:19} is not strictly a special case of Theorem~\ref{thm:netSubExist}.
Moreover, \cite{FlJaJaTe:19} also allowed for frictions such as transaction taxes and commissions in their existence result.}
However, \cite{FlJaJaTe:19} did not offer a maximal domain result for \gsubstity.
In the next section, we show that \gsubstity\ does not actually drive existence of \ce\ with substitutable indivisible goods.

\subsection{\NSubstity\ and the Existence of \CE}
\label{sec:nsub}

In light of Theorem~\ref{thm:existDualExchange} and Fact~\ref{fac:subExist}, \ce\ exists if agents' Hicksian demands satisfy an appropriate \substity\ condition---i.e., if preferences satisfy a \emph{net} analogue of \substity.

We build on Definition~\ref{def:gsub} to define a concept of \nsubstity\ for settings with indivisibilities.
\Nsubstity\ is a version of the \nsubstity\ condition from classical consumer theory.
It requires that \emph{compensated} increases in the price of a good (i.e., price increases that are offset by compensating transfers) weakly raise demand for all other goods.

\begin{definition}[\NSubstity]
\label{def:nsub}
Suppose that agent $j$ demands at most one unit of each good.
A utility function $\utilFn{j}$ is a \emph{\nsubst\ utility function} if for all utility levels $\ub$, price vectors $\p$, and $\lambda > 0,$ whenever $\dH{j}{\p}{\ub} = \{\bun\}$ and $\dH{j}{\p + \lambda \e{i}}{\ub} = \{\bunpr\}$, we have that $\bunprComp{k} \ge \bunComp{k}$ for all goods $k \not= i$.
\end{definition}

For quasilinear utility functions, \nsubstity\ coincides with (gross) \substity. 
More generally, \nsubstity\ can be expressed as a condition on \dualvals.

\begin{remark}
\label{rem:netSubDual}
By Lemma~\ref{lem:dHvalH}, if an agent demands at most one unit of each good, then she has a \nsubst\ utility function if and only if her \dualvals\ at all utility levels are \subst\ valuations.
\end{remark}

We can apply Fact~\ref{fac:dualPrefs} and Remark~\ref{rem:netSubDual} to construct large classes of \nsubst\ preferences with income effects from families of \subst\ valuations.
There are several rich families of \subst\ valuations, including endowed assignment valuations \citep{HaMi:05} and matroid-based valuations \citep{ostrovsky2015gross}.
This leads to a large class of quasilogarithmic \nsubst\ utility functions.

\begin{example}[Example~\ref{eg:quasilog} continued]
\label{eg:quasilogSubs}
A quasilogarithmic utility function $\utilFn{j}$ is a \nsubst\ utility function if and only if the quasivaluation $\quasivalFn{j}$ is a \subst\ valuation.\footnote{Indeed, recall that Example~\ref{eg:quasilogDualVal} tells us that agent $j$'s \dualval\ at each utility level is a positive linear transformation of $\quasivalFn{j}$.
The conclusion follows by Remark~\ref{rem:netSubDual}.}
\end{example}

More generally, in light of Fact~\ref{fac:dualPrefs} and Remark~\ref{rem:netSubDual}, each family of \subst\ valuations leads to a class of \nsubst\ utility functions with income effects consisting of the utility functions whose \dualvals\ all belong to the family.
These classes are defined by conditions on substitution effects and do not restrict income effects.
By contrast, \gsubstity\ places substantial restrictions on the form of income effects.\footnote{See Remark E.1 in the Supplemental Material of \cite{FlJaJaTe:19}).}

To understand the difference between gross and \nsubstity, we compare the conditions in a setting in which agents have unit demand for goods.

\begin{example}[Example~\ref{eg:house} continued]
\label{eg:houseGrossNet}
Consider an agent, Martine, who owns a house $i_1$ and is considering selling it to purchase (at most) one of houses $i_2$ and $i_3$.
If Martine experiences income effects,
then her choice between $i_2$ and $i_3$ generally depends on the price she is able to procure for her house $i_1$.
For example, if $i_3$ is a more luxurious house than $i_2,$ then Martine may only demand $i_3$ if the value of her endowment is sufficiently large---i.e., if the price of her house $i_1$ is sufficiently high.
As a result, when Martine is endowed with $i_1$, she does not generally have \gsubst\ preferences: increases in the price of $i_1$ can lower Martine's demand for $i_2$.
That is, Martine can regard $i_2$ as a gross complement for $i_1$.
In contrast, Martine has \nsubst\ preferences---no compensated increase in the price of $i_1$ could make Martine stop demanding $i_2$---a condition that holds generally in the housing market economy.\footnote{\citet[Example 2]{DaKoMu:01} also showed the connection between \citeposs{quinzii1984core} housing market economy and a \substity\ condition, but formulated their discussion in terms of the shape of the convex hull at domains at which demand is multi-valued instead of \nsubstity.  Their discussion is equivalent to ours by Corollary 5 in \cite*{DaKoLa:2003} and Remark~\ref{rem:netSubDual}.}
Note also that, unlike \nsubstity, \gsubstity\ generally depends on endowments: if Martine were not endowed a house, she would have \gsubst\ preferences \citep{kaneko1982central,kaneko1983housing,DeGa:85}.
\end{example}

While Example~\ref{eg:houseGrossNet} shows that \nsubstity\ does not imply \gsubstity, it turns out that \gsubstity\ implies \nsubstity.

\begin{proposition}
\label{prop:ssubst}
If agent $j$ demands at most one unit of each good and there exists an endowment $\bundow$ of goods at which $\utilFn{j}$ is a \gsubst\ utility function, then $\utilFn{j}$ is \nsubst\ utility function.
\end{proposition}

Proposition~\ref{prop:ssubst} and Example~\ref{eg:houseGrossNet} show that \gsubstity\ (at any one endowment of goods) implies \nsubstity\ but places additional restrictions on income effects.
Nevertheless, the restrictions on substitution effects alone, entailed by \nsubstity, are sufficient for the existence of \ce.

\begin{theorem}
\label{thm:netSubExist}
If all agents demand at most one unit of each good and have \nsubst\ utility functions, then \ces\ exist for all \dowallocs.
\end{theorem}

Theorem~\ref{thm:netSubExist} is an immediate consequence of the Equilibrium Existence Duality and the existence of \ces\ in transferable utility economies under \substity.

\begin{proof}
Remark~\ref{rem:netSubDual} implies that the agents' \dualvals\ at all utility levels are \subst\ valuations.
Hence, Fact~\ref{fac:subExist} implies that \ces\ exist in the \dualecons\ for all profiles of utility levels if \andowalloc\ exists.
The theorem follows by the ``if'' direction of Theorem~\ref{thm:existDualExchange}.
\end{proof}

As \gsubstity\ implies \nsubstity\ (Proposition~\ref{prop:ssubst}), the existence of \ce\ under \gsubstity\ is a special case of Theorem~\ref{thm:netSubExist}.
But Theorem~\ref{thm:netSubExist} is more general: as Example~\ref{eg:houseGrossNet} shows, \nsubstity\ allows for forms of gross complementarities between goods, in addition to \gsubstity. 
%
%
The following example illustrates how the distinction between \gsubstity\ and \nsubstity\ relates to the existence of \ce\ when agents caan demand multiple goods.

\begin{example}[\GSubstity\ versus \NSubstity\ and the Existence of \CE]
\label{eg:grossVersusNetNumerical}
There are two goods and the total endowment is $\tot = (1,1)$.
There are two agents, which we call $j$ and $k,$ and $j$'s feasible set of consumption bundles of goods is $\Feas{j} = \{0,1\}^2$.

We consider the price vectors $\p = (2,2)$ and $\ppr = (4,2)$ and consider two examples in which agent $j$'s Marshallian demand changes from $(1,1)$ to $(0,0)$ as prices change from $\p$ to $\ppr$---a \gcompity.
But the consequences for the existence of \ce\ are different across the two cases.
In Case~\ref{eg:numericalComp}, the \gcompity\ reflects a net complementarity for $j$, and \ce\ may not exist if $k$ sees goods as net substitutes.
In Case~\ref{eg:numericalLog}, the \gcompity\ reflects only an income effect for $j$, as in Example~\ref{eg:houseGrossNet}, so \ce\ is guaranteed to exist if $k$ sees goods as net substitutes.
\begin{enumerate}[label=(\alph*),wide]
\item \label{eg:numericalComp} Suppose that $j$ has a quasilinear utility function with valuation given by
\[\val{j}{\bun} = \begin{cases}
0 & \text{if } \bun = (0,0),(0,1),(1,0)\\
5 & \text{if } \bun = (1,1).
\end{cases}\]
Here, $\valFn{j}$ is not a \subst\ valuation because $\dQL{j}{\p} = \{(1,1)\}$ while $\dQL{j}{\ppr} = \{(0,0)\}$: i.e., increasing the price of the first good can lower $j$'s demand for the second good.
If $\Feas{k} = \{(0,0),(0,1),(1,0)\}$ and agent $k$ has a quasilinear utility function with a \subst\ valuation given by
\begin{equation}
\label{eq:kUnitDem}
\val{k}{\bun} = \begin{cases}
0 & \text{if } \bun = (0,0)\\
4 & \text{if } \bun = (1,0)\\
3 & \text{if } \bun = (0,1),
\end{cases}
\end{equation}
then no \ce\ exists.\footnote{The existence of a feasible set of bundles of goods and a \subst\ valuation for $k$ for which no \ce\ exists follows from Fact~\ref{fac:subMaxDomain}.  To check that $\valFn{k}$ is an example of such a valuation, suppose, for sake of deriving a contradiction, that $(\bunj,\bunag{k})$ is the allocation of goods in a \ce.
The First Welfare Theorem implies that $\bunj = (1,1)$ and that $\bunag{k} = (0,0)$.
But for agent $j$ to demand $(1,1),$ the equilibrium prices would have to sum to at most 5, while for agent $k$ to demand $(0,0)$, the equilibrium prices would both have to be at least 3---a contradiction.
Hence, we can conclude that no \ce\ exists.}
\item \label{eg:numericalLog} Suppose instead that $\utilFn{j}$ is quasilogarithmic (as defined in Example~\ref{eg:quasilog}) with quasivaluation given by
\[\quasival{j}{\bun} = \begin{cases}
-11 & \text{if } \bun = (0,0)\\
-7 & \text{if } \bun = (0,1)\\
-4 & \text{if } \bun = (1,0)\\
-1 & \text{if } \bun = (1,1).
\end{cases}\]
At the endowment $\bundowj = (0,1)$ of goods, $\utilFn{j}$ is not a \gsubst\ utility function as, letting $\numerdowj = 3,$ we have that $\dM{j}{\p}{\bunndowj} = \{(1,1)\}$ while $\dM{j}{\ppr}{\bunndowj} = \{(0,0)\}$.\footnote{\label{fn:evalU}To show this, note that $\numerdowj - \ppr \cdot ((1,1) - \bundowj) = -1,$ so it would violate $j$'s budget constraint to demand $(1,1)$ at the price vector $\ppr$.
For the other bundles, note that
\[\begin{array}{c|c|c|c|c}
\bun & (0,0) & (0,1) & (1,0) & (1,1)\\ \hline
\util{j}{\numerdowj - \p \cdot (\bun - \bundowj),\bun} & \log \frac{5}{11} & \log \frac{3}{7} & \log \frac{3}{4} & \log 1\\ \hline
\util{j}{\numerdowj - \ppr \cdot (\bun - \bundowj),\bun} & \log \frac{5}{11} & \log \frac{3}{7} & \log \frac{1}{4} & \text{undef.,}
\end{array}\]
so $\dM{j}{\p}{\bunndowj} = \{(1,1)\}$ and $\dM{j}{\ppr}{\bunndowj} = \{(0,0)\}$.
}
That is, increasing the price of the first good can lower $j$'s Marshallian demand for the second good.
%
By contrast, as $\quasivalFn{j}$ is a \subst\ valuation, Example~\ref{eg:quasilogSubs} implies that
$\utilFn{j}$ is a \nsubst\ utility function: the \gcompity\ is entirely due to an income effect.
For example, at the utility level
\[\ub = \max_{\bunn \in \Feans{j} \mid \pprall \cdot \bunn \le \pprall \cdot \bunndowj} \util{j}{\bunn} = \log \frac{5}{11},\]
we have that $\dH{j}{\p}{u} = \{(1,0)\}$ and that $\dH{j}{\ppr}{\ub} = \{(0,0)\}$,\footnote{The expressions for $\dH{j}{\p}{\ub}$ and $\dH{j}{\ppr}{\ub}$ hold because agent $j$'s Hicksian valuation at utility level $\ub$ is $\frac{5}{11}$ times the quasivaluation $\quasivalFn{j}$ (by Example~\ref{eg:quasilogDualVal}).
%
}
so the decrease in the Marshallian demand for the second good as prices change from $\p$ to $\ppr$ at the endowment $\bunndowj$ reflects an income effect.
By Theorem~\ref{thm:netSubExist}, \ce\ exists whenever $k$ has a \nsubst\ utility function.
For example, if $k$ has a quasilinear utility function with a \subst\ valuation given by Equation (\ref{eq:kUnitDem}), then for the \dowalloc\ defined by $\bundowj = (0,1)$, $\bundowag{k} = (1,0)$, and $\numerdowj = \numerdowag{k} = 3,$ the price vector $(3,2)$ and the allocation of goods defined by $\bunj = (1,0)$ and $\bunag{k} = (0,1)$ comprise a \ce.\footnote{To show this, let $\hp = (3,2)$.
It is clear that $(0,1) \in \dQL{k}{\hp}$.
It remains to show that $(1,0) \in \dM{j}{\hp}{\bunndowj}$.
Note that $\numerdowj - \hp \cdot ((1,1) - \bundowj) = 0,$ so it would violate $j$'s budget constraint to demand $(1,1)$ at the price vector $\hp$.
For the other bundles, note that
\[
\util{j}{\numerdowj - \hp \cdot (\bun - \bundowj),\bun} = \begin{cases}
\log \frac{5}{11} & \text{if } \bun = (0,0)\\
\log \frac{3}{7} & \text{if } \bun = (0,1)\\
\log \frac{1}{2} & \text{if } \bun = (1,0),
\end{cases}
\]
so $\dM{j}{\hp}{\bunndowj} = \{(1,0)\}$.}
\end{enumerate}
\indent

In Case~\ref{eg:numericalLog}, agent $j$ has \nsubst\ preferences---leading to the guaranteed existence of \ce\ when agent $k$ has \nsubst\ preferences.
By contrast, in Case~\ref{eg:numericalComp}, agent $j$ does not have \nsubst\ preferences---and \ce\ may not exist when $k$ has \nsubst\ preferences.
\end{example}

In general, 
net substitutability forms a maximal domain for the existence of \ce. 
Specifically, if an agent does not have \nsubst\ preferences, then \ce\ may not exist when the other agents have \subst\ quasilinear preferences.

\begin{proposition}
\label{prop:netSubstMaxDomain}
Suppose that $\totComp{i} = 1$ for all goods $i$.
If $|J| \ge 2$, agent $j$ demands at most one unit of each good, and $\utilFn{j}$ is not a \nsubst\ utility function, then there exist sets $\Feas{k} \subseteq \{0,1\}^I$ of feasible bundles and \subst\ valuations $\valFn{k}: \Feas{k} \to \mathbb{R}$ for agents $k \not= j$, and \andowalloc\ for which no \ce\ exists.
\end{proposition}

Proposition~\ref{prop:netSubstMaxDomain} is an immediate consequence of the Equilibrium Existence Duality and the fact that \substity\ defines a maximal domain for the existence of \ce\ with transferable utility.

\begin{proof}
By Remark~\ref{rem:netSubDual}, there exists a utility level $\ub$ at which agent $j$'s \dualval\ $\valHDef{j}$ is not a \subst\ valuation.
Fact~\ref{fac:subMaxDomain} implies that there exist feasible sets $\Feas{k} \subseteq \{0,1\}^I$ and \subst\ valuations $\valFn{k}$ for agents $k \not= j$, for which \andowalloc\ exists but no \ce\ would exist with transferable utility if agent $j$'s valuation were $\valHDef{j}$.
With those sets $\Feas{k}$ of feasible bundles and valuations $\valFn{k}$ for agents $k \not= j,$ the ``only if'' direction of Theorem~\ref{thm:existDualExchange} implies that there exists \andowalloc\ for which no \ce\ exists.
\end{proof}

Proposition~\ref{prop:netSubstMaxDomain} entails that any domain of preferences that contains all \subst\ quasilinear preferences and guarantees the existence of \ce\ must lie within the domain of \nsubst\ preferences.
Therefore, Proposition~\ref{prop:netSubstMaxDomain} and Theorem~\ref{thm:netSubExist} suggest that \nsubstity\ is the most general way to incorporate income effects into a \substity\ condition to ensure the existence of \ce.

By contrast, the relationship between the nonexistence of \ce\ and failures of \gsubstity\ depends on why \gsubstity\ fails.
\Gsubstity\ can fail due to substitution effects that reflect net complementarities, as in Example~\ref{eg:grossVersusNetNumerical}\ref{eg:numericalComp}, or due to income effects, as in Example~\ref{eg:grossVersusNetNumerical}\ref{eg:numericalLog}.
If the failure of \gsubstity\ reflects a net complementarity, then Proposition~\ref{prop:netSubstMaxDomain} tells us that \ce\ may not exist if the other agents have \subst\ quasilinear preferences, as in Example~\ref{eg:grossVersusNetNumerical}\ref{eg:numericalComp}.
On the other hand, the failure of \gsubstity\ is only due to income effects, then Theorem~\ref{thm:netSubExist} tells us that \ce\ exists if the other agents have \nsubst\ preferences (e.g., \subst\ quasilinear preferences), as in Example~\ref{eg:grossVersusNetNumerical}\ref{eg:numericalLog}.





\section{Demand Types and the Unimodularity Theorem}
\label{sec:demTypes}

In this section, we characterize exactly what conditions on patterns of substitution effects guarantee the existence of \ce.
Specifically, we consider \citeposs{BaKl:19} classification of valuations into ``demand types'' based on sets of vectors that summarize the possible ways in which demand can change in response to a small generic price change.
We first review the definition of demand types from \cite{BaKl:19}.
We then extend the concept of demand types to settings with income effects, and develop a version of the \citeposs{BaKl:19} Unimodularity Theorem that allows for income effects and characterizes which demand types guarantee the existence of \ce\ (see also \cite{DaKoMu:01}).
A special case of the Unimodularity Theorem with Income Effects extends Theorem~\ref{thm:netSubExist} to settings in which agents can demand multiple units of some goods.

\subsection{Demand Types and the Unimodularity Theorem with Transferable Utility}
\label{sec:demTypesTU}
We first review the concept of demand types for quasilinear settings, as developed by \citet{BaKl:19}.

An integer vector is \emph{primitive} if the greatest common divisor of its components is 1.
By focusing on the directions of demand changes, we can restrict to primitive demand change vectors.
A \emph{\dtvs} is a set $\mathcal{D} \subseteq \mathbb{Z}^I$ of primitive \intvecs\ such that if $\dvec \in\mathcal{D}$ then $- \dvec \in\mathcal{D}$.

\begin{definition}[Demand Types for Valuations]
\label{def:demTypeTU}
Let $\valFn{j}$ be a valuation.
\begin{enumerate}[label=(\alph*)]
\item A bundle $\bun$ is \emph{uniquely demanded by agent $j$} if there exists a price vector $\p$ such that $\dQL{j}{\p} = \{\bun\}.$
\item A pair $\{\bun,\bunpr\}$ of uniquely demanded bundles are \emph{adjacently demanded by agent $j$} if there exists a price vector $\p$ such that $\dQL{j}{\p}$ contains $\bun$ and $\bunpr$ but no other bundle that is uniquely demanded by agent $j$.
\item If $\mathcal{D}$ is a \dtvs, then $\valFn{j}$ is \emph{of demand type $\mathcal{D}$} if for all pairs $\{\bun,\bunpr\}$ that are adjacently demanded by agent $j$, the difference $\bunpr - \bun$ is a multiple of an element of $\mathcal{D}$.\footnote{Definition~\ref{def:demTypeTU}(c) coincides with Definition 3.1 in \cite{BaKl:19} by Proposition 2.20 in \cite{BaKl:19}.}
\end{enumerate}
\end{definition}
  
\begin{figure}
\centering
\def\svgwidth{0.4\columnwidth}
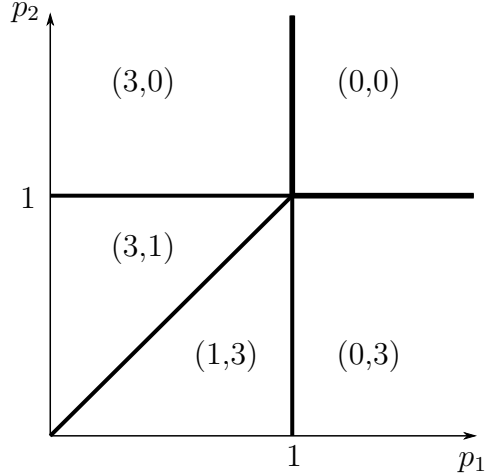
\caption{Depiction of Agent $j$'s Demand in Example~\ref{eg:demTypeDef}.
The labels indicate demand in the regions of price vectors at which demand is single-valued.}
\label{fig:demTypeDef}
\end{figure}

For intuition, suppose that a small price change causes a change in demand.  Then, generically, demand changes between adjacently demanded bundles.  Thus, the demand type vectors represent the possible directions of changes in demand in response to small generic price changes (see Proposition 3.3 in \cite{BaKl:19}  for a formal statement).  
To illustrate Definition~\ref{def:demTypeTU}, we consider an example.

\begin{example}
\label{eg:demTypeDef}
Suppose that there are two goods and let
\[\Feas{j} = \{0,1,2,3\}^2 \ssm \{(2,3),(3,2),(3,3)\}.\]
Consider the valuation defined by $\val{j}{\bun} = \bunComp{1} + \bunComp{2}$.
As Figure~\ref{fig:demTypeDef} illustrates, the uniquely demanded bundles are $(0,0)$, $(0,3),$ $(1,3),$ $(3,0)$, and $(3,1)$.

When $1 = \pComp{1} < \pComp{2},$ agent $j$'s demand is $\dQL{j}{\p} = \{(0,0),(1,0),(2,0),(3,0)\}$.
Hence, as the bundles $(1,0)$ and $(2,0)$ are not uniquely demanded, the bundles $(0,0)$ and $(3,0)$ are adjacently demanded.
As a result, for $\valFn{j}$ to be of demand type $\mathcal{D}$, the set $\mathcal{D}$ must contain the  vector $(1,0)$, which is the primitive \intvec\ proportional to the demand change $(3,0) - (0,0) = (3,0)$.
Similarly, the bundles $(0,0)$ and $(0,3)$ are adjacently demanded, and any \dtvs\ $\mathcal{D}$ such that $\valFn{j}$ is of demand type $\mathcal{D}$ must contain the vector $(0,1)$.

When $\pComp{1} < \pComp{2} = 1,$ demand is $\dQL{j}{\p} = \{(3,0),(3,1)\}$.
Hence, the bundles $(3,0)$ and $(3,1)$ are adjacently demanded.
Similarly, the bundles $(0,3)$ and $(1,3)$ are adjacently demanded.  These facts respectively imply, again, that $(0,1)$ and $(1,0)$ are in any \dtvs\ $\mathcal{D}$ such that $V^j$ is of demand type $\mathcal{D}$.

Last, when $p_1 = p_2 < 1,$ agent $j$'s demand is $\dQL{j}{\p} = \{(1,3),(2,2),(3,1)\}$.
Hence, as the bundle $(2,2)$ is not uniquely demanded, the bundles $(1,3)$ and $(3,1)$ are adjacently demanded.
As a result, for $\valFn{j}$ to be of demand type $\mathcal{D}$, the set $\mathcal{D}$ must contain the  vector $(1,-1)$, which is the primitive \intvec\ proportional to the demand change $(3,1) - (1,3) = (2,-2)$.

By contrast, the bundles $(0,0)$ and $(3,1)$ are not adjacently demanded: the only price vector at which agent $j$ demands them both is $\p = (1,1),$ but $\dQL{j}{1,1}$ also contains the uniquely demanded bundles $(0,3)$, $(1,3),$ and $(3,0)$.  Similarly, the bundles $(0,0)$ and $(1,3)$ are not adjacently demanded.
Hence,
\begin{equation*}
\mathcal{D} = \pm \left\{\begin{bmatrix}1 \\ 0\end{bmatrix},
\begin{bmatrix}0 \\ 1\end{bmatrix},
\begin{bmatrix} 1 \\ -1\end{bmatrix}\right\}
\end{equation*}
is the minimal \dtvs\ $\mathcal{D}$ such that $\valFn{j}$ is of demand type $\mathcal{D}$.
\end{example}


Consider any valuation of the same demand type $\mathcal{D}$ as in Example \ref{eg:demTypeDef}, and a change in price from  $\p$ to $\ppr=\p+\lambda \e{1}$ for some $\lambda>0$.  For generic choices of $\p$ and $\lambda$, the demand at any price on the straight line from $\mathbf{p}_I$ to $\mathbf{p}_I'$ either is unique, or demonstrates the adjacency of two bundles uniquely demanded at prices on this line.  The change in demand between such bundles must therefore be a multiple of an element of $\mathcal{D}$ (by Definition \ref{def:demTypeTU}).  Moreover, since only the price of good 1 is changing and that price is increasing, the law of demand entails that demand for good 1 must strictly decrease upon any change in demand.\footnote{As there are no income effects here, the compensated law of demand (see, e.g., Proposition 3.E.4 in  \cite{MaWhGr:95}) reduces to the law of demand.} 
Thus, the change in demand between the two consecutive uniquely demanded bundles must be a positive multiple of either $(-1,0)$ or $(-1,1)$.  Therefore, demand for good 2 must (weakly) increase, reflecting \substity\ between the goods.  This two-good example is a special case of an important class of demand types.

\begin{example}[The Strong Substitutes Demand Type]
\label{eg:ssubDemType}
The \emph{strong substitutes \dtvs} consists of all
vectors in $\Z^I$ with at most one $+1$ component, at most one $-1$ component, and no other nonzero components.
As illustrated in Example~\ref{eg:demTypeDef}, this \dtvs\ captures one-to-one substitution between goods through demand type vectors with one component of $1$ and one component of $-1$.
Furthermore, if an agent $k$ demands at most one unit of each good, then $\valFn{k}$ is a \subst\ valuation if and only if it is of the strong substitutes demand type (see Theorems 2.1 and 2.4 in \cite{fujishige2003note}).
\end{example}

In settings in which agents can demand multiple units of each good, a form of concavity is needed to ensure the existence of \ce.
A valuation is concave if, under that valuation, each bundle of goods that is a convex combination of feasible bundles of goods is demanded at some price vector.
For the formal definition, we let $\conv(T)$ denote the \emph{convex hull} of a set $T \subseteq \mathbb{R}^I.$

\begin{definition}[Concavity]\label{def:concave}
A valuation $\valFn{j}$ is \emph{concave} if for each bundle $\bun \in \conv(\Feas{j}) \cap \mathbb{Z}^n,$ there exists a price vector $\p$ such that $\bun \in \dQL{j}{\p}$.
\end{definition}

In Section~\ref{sec:subOld}, we discussed that \substity\ guarantees the existence of \ce\ in transferable utility economies when agents demand at most one unit of each good.
Generalizing that result, \citet{BaKl:19} identified a necessary and sufficient condition for the concave valuations of a demand type to form a domain for the guaranteed existence of \ce.

\begin{definition}[Unimodularity]
A set of vectors in $\Z^I$ is \emph{unimodular} if every linearly independent subset can be extended to be a basis for $\R^I$, of \intvecs, such that any square matrix whose columns are these vectors has determinant $\pm 1$.
\end{definition}

For example, the \dtvs\ in Example~\ref{eg:demTypeDef} is unimodular, while the \dtvs\ 
\begin{equation}
\label{eq:subsComp}
\pm \left\{\begin{bmatrix} 1 \\ -1\end{bmatrix},
\begin{bmatrix}1 \\ 1\end{bmatrix}\right\}
\end{equation}
is not unimodular, because
\[\left|\begin{matrix}
1 & 1\\
-1 & 1
\end{matrix}\right| = 2.\]
The \dtvs\ in (\ref{eq:subsComp}) represents that the two goods can be substitutable or complementary for agents---a possibility that can cause \ce\ to fail to exist, as in Example~\ref{eg:grossVersusNetNumerical}\ref{eg:numericalComp}.
\cite{BaKl:19} showed that the unimodularity of a \dtvs\ is precisely the condition for the corresponding demand type to guarantee the existence of \ce.

\begin{fact}[Unimodularity Theorem with Transferable Utility]
\label{fac:unimod}
Let $\mathcal{D}$ be a \dtvs.
\Ces\ exist for all finite sets $J$ of agents with concave valuations of demand type $\mathcal{D}$ and for all total endowments for which \dowallocs\ exist if and only if $\mathcal{D}$ is unimodular.\footnote{\label{fn:classDiscConv}The ``if'' direction of Fact~\ref{fac:unimod} is a case of the ``if" direction of Theorem 4.3 in \cite{BaKl:19}.
The ``only if'' direction of Fact~\ref{fac:unimod}, which we prove in Appendix~\ref{app:maxDomain}, is a mild strengthening of the ``only if" direction of Theorem 4.3 in \cite{BaKl:19} that applies in exchange economies.}
\end{fact}

\cite{DaKoMu:01} used conditions on the ranges of agents' demand correspondences to describe classes of concave valuations, which correspond to the concave valuations of \citeposs{BaKl:19} unimodular demand types;\footnote{To understand the correspondence, let $\mathcal{D}$ be a unimodular \dtvs. In the terminology of \cite{DaKoMu:01}, a valuation $\valFn{j}$ is \emph{$\mathscr{D}(\mathscr{P}t(\mathcal{D},\mathbb{Z}))$-concave} if, for each price vector $\p$, we have that $\dQL{j}{\p} = \conv(\dQL{j}{\p}) \cap \Z^I$ and each edge of $\conv(\dQL{j}{\p})$ is parallel to an element of $\mathcal{D}$ (see Definition 4 and pages 264--265 in \cite{DaKoMu:01}). It follows from Lemma 2.11 and Proposition 2.16 in \cite{BaKl:19} that a valuation is $\mathscr{D}(\mathscr{P}t(\mathcal{D},\mathbb{Z}))$-concave if and only if it is concave and of demand type $\mathcal{D}$.} they formulated a version of the ``if" direction of Fact~\ref{fac:unimod} with those conditions.\footnote{See Definition 4, Theorem 3, and pages 264--265 in \cite{DaKoMu:01}.}

As \cite{poincare1900second} showed, the strong substitutes \dtvs\ is unimodular.
Therefore, in light of Example~\ref{eg:ssubDemType}, the existence of \ce\ in transferable utility economies in which agents demand at most one unit of each good and have substitutes valuations (Fact~\ref{fac:subExist}) is a special case of Fact~\ref{fac:unimod}.
Moreover, Fact~\ref{fac:unimod} is strictly more general:
as \cite{BaKl:19} showed, there are unimodular \dtvs{s} for which the existence of \ce\ cannot be deduced from the corresponding result for strong substitutes by applying a change of basis to the space of bundles of goods.\footnote{By contrast, the existence results of \cite{SuYa:06} and \cite{Teyt:14} can be deduced from Fact~\ref{fac:subExist} applying an appropriate change of basis.  Those results are also special cases of Fact~\ref{fac:unimod}.}
To illustrate the additional generality, we discuss an example of such a demand type.\footnote{Section 6.1 in \cite{BaKl:19} provides another example that includes only complements valuations.}

\begin{example}
\label{eg:5D}
There are five goods.
Consider the \dtvs\
\[\mathcal{D} = \pm\left\{\begin{bmatrix}1 \\ 0 \\ 0 \\ 0 \\ 0\end{bmatrix},
\begin{bmatrix}0 \\ 1 \\ 0 \\ 0 \\ 0\end{bmatrix},
\begin{bmatrix}0 \\ 0 \\ 1 \\ 0 \\ 0\end{bmatrix},
\begin{bmatrix}0 \\ 0 \\ 0 \\ 1 \\ 0\end{bmatrix},
\begin{bmatrix}0 \\ 0 \\ 0 \\ 0 \\ 1\end{bmatrix},
\begin{bmatrix}1 \\ -1 \\ 1 \\ 0 \\ 0\end{bmatrix},
\begin{bmatrix}0 \\ 1 \\ -1 \\ 1 \\ 0\end{bmatrix},
\begin{bmatrix}0 \\ 0 \\ 1 \\ -1 \\ 1\end{bmatrix},
\begin{bmatrix}1 \\ 0 \\ 0 \\ 1 \\ -1\end{bmatrix},
\begin{bmatrix}-1 \\ 1 \\ 0 \\ 0 \\ 1\end{bmatrix}\right\}.\]
Intuitively, this \dtvs\ allows for independent changes in the demand for each good (through the first five vectors), as well as for substitution from a good to the bundle consisting of its two neighbors if the goods are arranged in a circle (through the last five vectors).
This \dtvs\ is unimodular, and cannot be obtained from the strong substitutes \dtvs\ by a change of basis of the space of integer bundles of goods (see, e.g., Section 19.4 of \cite{schrijver1998theory}).
\end{example}

Moreover, the demand types defined by maximal, unimodular \dtvs{s} turn out to define maximal domains for the existence of \ce\ in settings with transferable utility.
Here, we say that a unimodular \dtvs\ is \emph{maximal} if it is not strictly contained in another unimodular \dtvs.

\begin{fact}
\label{fac:unimodMaxDomain}
Let $\mathcal{D}$ be a maximal unimodular \dtvs.
If $|J| \ge 2$ and $\valFn{j}$ is non-concave or not of demand type $\mathcal{D}$, then there exist sets $\Feas{k}$ of feasible bundles and concave valuations $\valFn{k}: \Feas{k} \to \mathbb{R}$ of demand type $\mathcal{D}$ for agents $k \not= j$, as well as a total endowment, for which there exists \andowalloc\ but no \ce.\footnote{Fact~\ref{fac:unimodMaxDomain} is related to Proposition 6.10 in \cite{BaKl:14}, which connects failures of unimodularity to the non-existence of \ce\ in specific economies.
We supply a proof of Fact~\ref{fac:unimodMaxDomain} in Appendix~\ref{app:maxDomain}.}
\end{fact}

While Fact~\ref{fac:unimod} shows that there exist valuations in each non-unimodular demand type for which \ce\ does not exist, Fact~\ref{fac:unimodMaxDomain} shows that for \emph{every} valuation outside a maximal unimodular demand type, there exist concave valuations within the demand type that lead to non-existence.
Hence, the necessity direction of Fact~\ref{fac:unimod}, together with Fact~\ref{fac:unimodMaxDomain}, provide complementary perspectives on the way in which \ce\ can fail to exist outside the context of unimodular demand types.

\subsection{Demand Types and the Unimodularity Theorem with Income Effects}

We now use Fact~\ref{fac:dualPrefs} to extend the demand types framework to settings with income effects.

\begin{definition}[Demand Types with Income Effects]
\label{def:demTypeIncEff}
An agent's preferences are \emph{of demand type $\mathcal{D}$} if her \dualvals\ at all utility levels are of demand type $\mathcal{D}$.
\end{definition}

Lemma~\ref{lem:dHvalH} leads to an economic interpretation of Definition~\ref{def:demTypeIncEff}: a utility function is of demand type $\mathcal{D}$ if $\mathcal{D}$ summarizes the possible ways in which Hicksian demand can change in response to a small generic price change.
In particular, Definition~\ref{def:demTypeIncEff} extends the concept of demand types to settings with income effects by placing conditions on substitution effects.
Indeed, Definition~\ref{def:demTypeIncEff} considers only the properties of \dualvals\ at each utility level (which, by Lemma~\ref{lem:dHvalH}, reflect substitution effects), and not how an agent's \dualvals\ vary with her utility level (which, by Fact~\ref{fac:dualDem} and Lemma~\ref{lem:dHvalH}, reflects income effects).

\cite{DaKoMu:01} translated their conditions on the ranges of agents' demand correspondences from quasilinear settings to settings with income effects by using Fact~\ref{fac:dualPrefs} in an analogous manner (see Assumption 3$'$ in \cite{DaKoMu:01}).
However, the economic interpretation in terms of substitution effects that Lemma~\ref{lem:dHvalH} leads to was not clear from \citeposs{DaKoMu:01} formulation.

As with the case of transferable utility, a concavity condition is needed to ensure the existence of \ce.
With income effects, the relevant condition is a version of the quasiconcavity condition from classical demand theory for settings with indivisible goods.  We define quasiconcavity based on concavity and duality.\footnote{It is equivalent to define quasiconcavity in terms of the convexity of the upper contour sets, but Definition \ref{def:quasiConc} is more immediately applicable for us.}

\begin{definition}[Quasiconcavity]
\label{def:quasiConc}
An agent's utility function is \emph{quasiconcave} if her \dualvals\ at all utility levels are concave.
\end{definition}

As with the case of transferable utility, unimodularity is a necessary and sufficient condition for the existence of \ce\ to be guaranteed for all quasiconcave preferences of a demand type when income effects are present.

\begin{theorem}[Unimodularity Theorem with Income Effects]
\label{thm:unimod}
Let $\mathcal{D}$ be a \dtvs.
\Ces\ exist for all finite sets $J$ of agents with quasiconcave utility functions of demand type $\mathcal{D}$, for all total endowments, and for all \dowallocs\ if and only if $\mathcal{D}$ is unimodular.
\end{theorem}

The ``only if'' direction of Theorem~\ref{thm:unimod} is a special case of the Unimodularity Theorem with Transferable Utility (Fact~\ref{fac:unimod}).
The ``if'' direction of Theorem~\ref{thm:unimod} is an immediate consequence of the Equilibrium Existence Duality and Fact~\ref{fac:unimod}.

\begin{proof}[Proof of the ``if" direction of Theorem~\ref{thm:unimod}]
Consider a finite set $J$ of agents with quasiconcave preferences of demand type $\mathcal{D}$ and a total endowment for which \andowalloc\ exists.
By definition, the agents' \dualval{s} at all utility levels are concave and of demand type $\mathcal{D}$.
Hence, \ces\ exist in the \dualecons\ for all profiles of utility levels by the ``if'' direction of Fact~\ref{fac:unimod}.
By the ``if'' direction of Theorem~\ref{thm:existDualExchange}, \ces\ must therefore exist in the original economy for all \dowallocs.
%
\end{proof}

\cite{DaKoMu:01} proved a version of the ``if'' direction of Theorem~\ref{thm:unimod} under the assumptions that utility functions are monotone in goods, that consumption of goods is nonnegative, and that the total endowment is strictly positive (see Theorems 2 and 4 in \cite{DaKoMu:01}).\footnote{\citeposs{DaKoMu:01} existence result is not formally a special case of ours because they allowed for unbounded sets $\Feas{j}$ of feasible bundles of goods.}
Note that they formulated their result in terms of Fact~\ref{fac:dualPrefs} and a condition on the ranges of 
demand correspondences 
(see their Assumption 3$'$) instead of in terms of unimodular demand types.

\citeposs{DaKoMu:01} approach was to show the existence of \ce\ in a convexified economy and that, under unimodularity, \ces\ in the convexified economy give rise to \ces\ in the original economy.
In contrast, our approach of using the Equilibrium Existence Duality illuminates the role of substitution effects in ensuring the existence of \ce.
Moreover, it yields a maximal domain result for unimodular demand types with income effects.

\begin{proposition}
\label{prop:netUnimodMaxDomain}
Let $\mathcal{D}$ be a maximal unimodular \dtvs.
If $|J| \ge 2$ and $\utilFn{j}$ is not quasiconcave or not of demand type $\mathcal{D}$, then there exist sets $\Feas{k}$ of feasible bundles and concave valuations $\valFn{k}: \Feas{k} \to \mathbb{R}$ of demand type $\mathcal{D}$ for agents $k \not= j$, as well as a total endowment and \andowalloc, for which no \ce\ exists.
\end{proposition}

Proposition~\ref{prop:netUnimodMaxDomain} is an immediate consequence of the Equilibrium Existence Duality and the maximal domain result for unimodular demand types under the transferability of utility.
 
\begin{proof}
By definition, there exists a utility level $\ub$ at which agent $j$'s \dualval\ $\valHDef{j}$ is non-concave or not of demand type $\mathcal{D}$.
In either case, Fact~\ref{fac:unimodMaxDomain} implies that there exist sets $\Feas{k}$ of feasible bundles and concave valuations $\valFn{k}: \Feas{k} \to \mathbb{R}$ of demand type $\mathcal{D}$ for agents $k \not= j$, and a total endowment for which \andowalloc\ exists but no \ce\ would exist with transferable utility if agent $j$'s valuation were $\valHDef{j}$.
With those sets $\Feas{k}$ of feasible bundles and valuations $\valFn{k}$ for agents $k \not= j$ and that total endowment, the ``only if" direction of Theorem~\ref{thm:existDualExchange} implies that there exists \andowalloc\ for which no \ce\ exists.
\end{proof}

Intuitively, Proposition~\ref{prop:netUnimodMaxDomain} and Theorem~\ref{thm:unimod} suggest that Definition~\ref{def:demTypeIncEff} is the most general way to incorporate income effects into unimodular demand types from the quasilinear setting and ensure the existence of \ce.
Indeed, Proposition~\ref{prop:netUnimodMaxDomain} entails that any domain of preferences that contains all concave quasilinear preferences of a maximal, unimodular demand type and guarantees the existence of \ce\ must lie within the corresponding demand type constructed in Definition~\ref{def:demTypeIncEff}.



\subsection{The Strong Substitutes Demand Type and Net \Substity\ with Multiple Units}
\label{sec:ssub}

We now use the case of Theorem~\ref{thm:unimod} for the strong substitutes demand type to extend Theorem~\ref{thm:netSubExist} to settings in which agents can demand multiple units of some goods.
In such settings, if utility is transferable, the \substity\ condition needed to ensure the existence of \ce\ is \emph{\ssubstity}---the condition requiring that agents see units of goods as \subst\ \citep{MiSt:09}.
As \cite{shioura2015gross} and \cite{BaKl:19} showed, there is a close relationship between strong (net) \substity\ and the strong substitutes demand type.\footnote{Requiring that different goods, rather than different units of goods, be \subst\ leads to a condition called \emph{\osubstity}.  However, \osubstity\ does not ensure the existence of \ce\ when agents can demand multiple units of some goods \citep*{DaKoLa:2003,MiSt:09,BaKl:19}.  \Osubstity\ in turn corresponds to an ``ordinary substitutes'' demand type (see Definitions 3.4 and 3.5 and Proposition 3.6 in \cite{BaKl:19}).}

\begin{definition}[\SSubstity]
\begin{enumerate}[label=(\alph*)]
\item A valuation is a \emph{\ssubst\ valuation} if it corresponds to a \subst\ valuation when each unit of each good is regarded as a separate good.
\item A utility function is a \emph{\snsubst\ utility function} if it corresponds to a \nsubst\ utility function when each unit of each good is regarded as a separate good.
\end{enumerate}
\end{definition}

\begin{fact}
\label{fac:ssubDemTypeConc}
A valuation (resp.~utility function) is a strong (net) \subst\ valuation (resp. utility function) if and only if it is concave (resp.~quasiconcave) and of the strong substitutes demand type.\footnote{The quasilinear case of this fact is part of Theorem 4.1(i) in \cite{shioura2015gross} (see also Proposition 3.10 in \cite{BaKl:19}).
The general case follows from the quasilinear case by Lemma~\ref{lem:dHvalH} and Remark~\ref{rem:netSubDual}.}\footnote{In particular, if agent $j$ demands at most one unit of each good, then $\utilFn{j}$ is a \nsubst\ utility function if and only if it is of the strong substitutes demand type.}
\end{fact}

As the strong substitutes \dtvs\ is unimodular \citep{poincare1900second}, the existence of \ce\ under \snsubstity\ is therefore a special case of the Unimodularity Theorem with Income Effects.

\begin{corollary}
\label{cor:snsubExist}
If all agents have \snsubst\ utility functions, then \ces\ exist for all \dowallocs.
\end{corollary}

Corollary~\ref{cor:snsubExist} can also be proven directly using the Equilibrium Existence Duality and the existence of \ce\ under \ssubstity\ in transferable utility economies  \citep{MiSt:09,ikebe2015stability}.
Theorem~\ref{thm:netSubExist} is the special case of Corollary~\ref{cor:snsubExist} for settings in which agents demand at most one unit of each good.
As there are unimodular \dtvs{s} unrelated to the strong substitutes \dtvs\ (such as the one in Example~\ref{eg:5D}), Theorem~\ref{thm:unimod} is strictly more general than Corollary~\ref{cor:snsubExist} (and hence Theorem~\ref{thm:netSubExist}).
In particular, Theorem~\ref{thm:unimod} also illustrates that certain patterns of net complementarities can also be compatible with the existence of \ce.

As the strong substitutes \dtvs\ is maximal as a unimodular \dtvs\ (see, e.g., Example 9 in \cite{danilov2004discrete}), Proposition~\ref{prop:netUnimodMaxDomain} yields a maximal domain result for \snsubstity.

\begin{corollary}
\label{cor:snsubMax}
If $|J| \ge 2$ and $\utilFn{j}$ is not a \snsubst\ utility function, then there exist \ssubst\ valuations $\valFn{k}$ for agents $k \not= j$, as well as a total endowment and \andowalloc, for which no \ce\ exists.
\end{corollary}

\section{Auction Design}
\label{sec:auctions}

\defcitealias{baldwin2020implementing}{Baldwin and Klemperer (in preparation)}

Our work has several implications for auction design.
First, our perspective of analyzing preferences by using the expenditure\-/minimization problem may yield new approaches for extending auction bidding languages to allow for income effects.

Second, our equilibrium existence results suggest that some auctions with competitive equilibrium pricing may work well for indivisible goods even in the presence of financing constraints.
One set of examples are Product-Mix Auctions, such as the one implemented by the Bank of England\footnote{See \cite{klemperer2008new,klemperer2010product,klemperer2018product} and \citetalias{baldwin2020implementing}.  Iceland planned a Product-Mix Auction for bidders with budget constraints \citep{klemperer2018product}, but that auction was for a setting with divisible goods.}---these implement
competitive equilibrium allocations assuming that the submitted sealed bids represent bidders' actual preferences, since truth-telling is a reasonable approximation in these auctions when there are sufficiently many bidders.

However, while we have shown that gross complementarities do not lead to the nonexistence of competitive equilibrium, they do create problems for dynamic auctions.
When agents see goods as gross substitutes, iteratively increasing the prices of over-demanded goods leads to a competitive equilibrium \citep{KeCr:82,FlJaJaTe:19}.
In contrast, when there are gross complementarities between goods, increases in the price of an over-demanded good can lead to other goods being under-demanded due to an income effect.
So, even though competitive equilibrium always exists when agents see goods as (strong) net substitutes, it may not be possible to find a competitive equilibrium using a monotone, dynamic auction.
In particular, simple ``activity rules'' that require bidders to bid on a smaller total number of units of goods as prices increase may result in inefficient outcomes.
So, the Product-Mix Auction approach of finding competitive equilibrium based on a single round of sealed bids seems especially useful in the presence of income effects.

\section{Conclusion}
\label{sec:conclusion}

The Equilibrium Existence Duality is a useful tool for analyzing economies with indivisible goods.
It is based on the relationship between Marshallian and Hicksian demands, and on an interpretation of Hicksian demand in terms of a quasilinear maximization problem.
The Equilibrium Existence Duality shows that competitive equilibrium exists (for all endowment allocations) if and only if competitive equilibrium exists in each of a family of Hicksian economies.
An application is that it is net substitutability, not gross substitutability, that is relevant to the existence of equilibrium.
And extending the demand types classification of valuations \citep{BaKl:19} allows us to state a Unimodularity Theorem with Income Effects that gives conditions on the patterns of substitution effects that guarantee the existence of competitive equilibrium.
In short, with income effects, just as without them, existence does not depend on agents seeing goods as substitutes; rather, substitution effects are fundamental to the existence of competitive equilibrium.

Our results point to a number of potential directions for future work.
First, it would be interesting to investigate applications of the Equilibrium Existence Duality to other results on the existence of equilibrium with transferable utility---such as those of \cite{BiMa:97}, \cite{Ma:98}, and \cite{CaOzPa:15}.
Second, our results could be used to further develop auction designs that find competitive equilibrium outcomes given the submitted bids, such as \citeposs{klemperer2008new} Product-Mix Auction. More broadly, our approach may lead to new results about the properties of economies with indivisibilities and income effects.

\appendix
\myspacing

\makeatletter
\let\c@fact=\c@theorem
\let\c@lemma=\c@theorem
\let\c@proposition=\c@theorem
\let\c@corollary=\c@theorem
\let\c@definition=\c@theorem
\let\c@claim=\c@theorem
\let\c@remark=\c@theorem
\let\c@example=\c@theorem
\makeatother

\section{Proof of Theorem~\ref{thm:existDualExchange} and Lemma~\ref{lem:dualEconSWT}}
\label{app:EEDproof}

We prove the following result, which combines Theorem~\ref{thm:existDualExchange} and Lemma~\ref{lem:dualEconSWT}.

\begin{theorem}
\label{thm:existDualExchangeSWT}
Suppose that the total endowment and the sets of feasible bundles are such that \andowalloc\ exists.
The following are equivalent.
\begin{enumerate}[label=(\Roman*)]
\item \label{cond:marshall} \Ces\ exist for all \dowallocs.
\item \label{cond:SWT} For each Pareto-efficient allocation $(\bunnj)_{j \in J}$ with $\sum_{j \in J} \bunj = \tot$, there exists a price vector $\p$ such that $\bunnj \in \dM{j}{\p}{\bunnj}$ for all agents $j$.
\item \label{cond:hicks} \Ces\ exist in the \dualecons\ for all profiles of utility levels.
\end{enumerate}
\end{theorem}



The remainder of this appendix is devoted to the proof of Theorem~\ref{thm:existDualExchangeSWT}.

\subsection{Proof of the \ref{cond:marshall}\texorpdfstring{$\implies$}{ implies }\ref{cond:SWT} Implication in Theorem~\ref{thm:existDualExchangeSWT}}

The proof of this implication is essentially identical to the proof of Theorem 3 in \cite{maskin2008fundamental}.
Consider a Pareto-efficient allocation $(\bunnj)_{j \in J}$ with $\sum_{j \in J} \bunnj = \tot.$

Let agent $j$'s endowment be $\bunndowj = \bunnj$.
By Statement~\ref{cond:marshall} in the theorem, there exists a \ce, say consisting of the price vector $\p$ and the allocation $(\hbunj)_{j \in J}$ of goods.
By the definition of \ce, we have that $\hbunj \in \dM{j}{\p}{\bunj}$ for all agents $j$.
In particular, letting $\hnumerj = \numerj - \p \cdot (\hbunj - \bunj)$ for each agent $j,$ we have that $\sum_{j \in J} \hbunnj = \sum_{j \in J} \bunnj$ and that $\util{j}{\hbunnj} \ge \util{j}{\bunnj}$ for all agents $j$.
As the allocation $(\bunnj)_{j \in J}$ is Pareto-efficient, we must have that $\util{j}{\hbunnj} = \util{j}{\bunnj}$ for all agents $j$.
It follows that $\bunj \in \dM{j}{\p}{\bunnj}$ for all agents $j$---as desired.

\subsection{Proof of the \ref{cond:SWT}\texorpdfstring{$\implies$}{ implies }\ref{cond:hicks} Implication in Theorem~\ref{thm:existDualExchangeSWT}}

Let $(\ubj)_{j \in J}$ be a profile of utility levels. 
Consider any allocation $(\bunj)_{j \in J} \in \bigtimes_{j \in J} \Feas{j}$ of goods with $\sum_{j \in J} \bunj = \tot$ that minimizes
\[\sum_{j \in J} \cf{j}{\bunj}{\ubj}\]
over all allocations $(\hbunj)_{j \in J} \in \bigtimes_{j \in J} \Feas{j}$ of goods with $\sum_{j \in J} \hbunj = \tot.$
Such an allocation exists because each set $\Feas{j}$ is finite and \andowalloc\ exists.
For each agent $j,$ let $\numerj = \cf{j}{\bunj}{\ubj}$---so $\util{j}{\bunnj} = \ubj$.

\begin{claim}
\label{cl:dualPE}
The allocation $(\bunnj)_{j \in J}$ is Pareto-efficient.
\end{claim}
\begin{proof}
Consider any allocation $(\hbunnj)_{j \in J} \in \bigtimes_{j \in J} \Feans{j}$ with $\sum_{j \in J} \hbunj = \tot$, and $\util{j}{\hbunnj} \ge \util{j}{\bunnj} = \ubj$ for all agents $j$ with strict inequality for some $j = j_1$.
As $\cf{j}{\hbunj}{\cdot}$ is strictly increasing for each agent $j$, we must have that
\[\hnumerj = \cf{j}{\hbunj}{\util{j}{\hbunnj}} \ge \cf{j}{\hbunj}{\ubj}\] for all agents $j$ with strict inequality for $j = j_1$.
Hence, we must have that
\[\sum_{j \in J} \hnumerj 
	> \sum_{j \in J} \cf{j}{\hbunj}{\ubj} 
	\ge \sum_{j \in J} \cf{j}{\bunj}{\ubj} 
	= \sum_{j \in J} \numerj,
\]
where the second inequality follows from the definition of $(\bunj)_{j \in J}$, so the allocation $(\bunnj)_{j \in J}$ cannot be Pareto-dominated.
\end{proof}

By Claim~\ref{cl:dualPE} and Statement~\ref{cond:SWT} in the theorem, there exists a price vector $\p$ such that $\bunj \in \dM{j}{\p}{\bunnj}$ for all agents $j$.
Fact~\ref{fac:dualDem} implies that $\bunj \in \dH{j}{\p}{\ubj}$ for all agents $j$.
By Lemma~\ref{lem:dHvalH}, it follows that the price vector $\p$ and the allocation $\left(\bunj\right)_{j \in J}$ of goods comprise a \ce\ in the \dualecon\ for the profile $(\ubj)_{j \in J}$ of utility levels.

\subsection{Proof of the \ref{cond:hicks}\texorpdfstring{$\implies$}{ implies }\ref{cond:marshall} Implication in Theorem~\ref{thm:existDualExchangeSWT}}

Let $(\bunndowj)_{j \in J}$ be an endowment allocation.
For each agent $j,$ we define a utility level $\umin{j} = \util{j}{\bunndowj}$ and let
\begin{align*}
K^j &= \numerdowj - \min_{\bun \in \Feas{j}} \cf{j}{\bun}{\umin{j}},
\end{align*}
which is non-negative by construction.
Furthermore, let $K = 1 + \sum_{j \in J} K^j$ and let
\[\umax{j} = \max_{\bun \in \Feas{j}} \util{j}{\numerdowj + K, \bun}.\]

Given a profile $\ubvec = (\ubj)_{j \in J}$ of utility levels, let
\[\payoffs{\ubvec} = \left\{\left(\begin{array}{l}
\cf{j}{\bunj}{\ubj} - \numerdowj\\
+ \p \cdot (\bunj - \bundowj)
\end{array}\right)_{j \in J} \lgiv \begin{array}{l}
\left(\p,(\bunj)_{j \in J}\right) \text{ is a competitive}\\
\text{equilibrium in the \dualecon}\\
\text{for the profile } (\ubj)_{j \in J} \text{ of utility levels}
\end{array}\lgivend
\right\}\]
denote the set of profiles of net expenditures over all \ces\ in the \dualecon\ for the profile $(\ubj)_{j \in J}$ of utility levels.
As discussed in Section~\ref{sec:EED}, the strategy of the proof is to solve for a profile $\ubvec = (\ubj)_{j \in J}$ of utility levels such that $\zero \in \payoffs{\ubvec}$. 

We first show that the correspondence $\payoffFn: \bigtimes_{j \in J} [\umin{j},\umax{j}] \toto \mathbb{R}^J$ is upper hemicontinuous and has compact, convex values.
We then apply a topological fixed point argument to show that there exists a profile $\ubvec = (\ubj)_{j \in J} \in \bigtimes_{j \in J} [\umin{j},\umax{j}]$ of utility levels such that $\zero \in \payoffs{\ubvec}$.
We conclude the proof by constructing a \ce\ for the endowment allocation $(\bunndowj)_{j \in J}$ in the original economy from a \ce\ in the \dualecon\ for the profile $(\ubj)_{j \in J}$ of utility levels.

\subsubsection*{Proof of the Regularity Conditions for $\payoffFn$.}
We begin by proving that the correspondence  $\payoffFn: \bigtimes_{j \in J} [\umin{j},\umax{j}] \toto \mathbb{R}^J$ is upper hemicontinuous and has compact, convex values.
We actually give explicit bounds for the range of $\payoffFn$.
Let
\[\overline{M} = \max_{j \in J} \left\{\cf{j}{\bundowj}{\umax{j}} - \numerdowj\right\}\]
and let
\[\underline{M} =  \sum_{j \in J} \left(\min_{\bun \in \Feas{j}} \left\{\cf{j}{\bun}{\umin{j}}\right\} - \numerdowj\right) - (|J|-1) \overline{M}.\]

\begin{claim}
\label{cl:regT}
The correspondence $\payoffFn: \bigtimes_{j \in J} [\umin{j},\umax{j}] \toto \mathbb{R}^J$ is upper hemicontinuous and has compact, convex values and range contained in $[\underline{M},\overline{M}]^J$.
\end{claim}

The proof of Claim~\ref{cl:regT} uses the following technical description of $\payoffFn$.

\begin{claim}
\label{cl:hicksEqHelp}
Let $\ubvec = (\ubj)_{j \in J} \in \bigtimes_{j \in J} [\umin{j},\umax{j}]$ be a profile of utility levels and let $(\bunj)_{j \in J} \in \bigtimes_{j \in J} \Feas{j}$ be an allocation of goods with $\sum_{j \in J} \bunj = \tot$.
If $(\bunj)_{j \in J}$ minimizes
\[\sum_{j \in J} \cf{j}{\hbunj}{\ubj}\]
over all allocations $(\hbunj)_{j \in J} \in \bigtimes_{j \in J} \Feas{j}$ of goods with $\sum_{j \in J} \hbunj = \tot$,
then we have that
\[\payoffs{\ubvec} = \left\{\rgivend \left(\cf{j}{\bunj}{\ubj} - \numerdowj + \p \cdot (\bunj - \bundowj)\right)_{j \in J} \rgiv \p \in \mathcal{P}\right\},\]
where
\[\mathcal{P} = \left\{\p \lgiv
\cf{j}{\bunj}{\ubj} + \p \cdot \bunj \le \cf{j}{\bunpr}{\ubj} + \p \cdot \bunpr
\text{ for all } j \in J \text{ and } \bunpr \in \Feas{j}
\lgivend\right\}.\]
\end{claim}
\begin{proof}
By construction, we have that
\[\mathcal{P} = \left\{\p \lgiv \begin{array}{l}
\left(\p,(\bunnj)_{j \in J}\right) \text{ is a \ce\ in the}\\
\text{ \dualecon\ for the profile } (\ubj)_{j \in J} \text{ of utility levels}
\end{array}\lgivend
\right\}.\]
A standard lemma regarding \ces\ in transferable utility economies shows that in the \dualecon\ for the profile $(\ubj)_{j \in J}$ of utility levels, if $\left(\p,(\hbunnj)_{j \in J}\right)$ is a \ce
, then so is $\left(\p,(\bunj)_{j \in J}\right)$.
\footnote{The lemma is due to \citet[page 3]{shapley1964values}; see also \cite{BiMa:97} and \cite{HaKoNiOsWe:11}. \citet[Lemma 1]{jagadeesan2020lone} proved the lemma in a setting with multiple units that allows for non-monotone valuations.}
In this case, we have that
\[\cf{j}{\bunj}{\ubj} + \p \cdot \bunj = \cf{j}{\hbunj}{\ubj} + \p \cdot \hbunj,\]
and hence that
\[\cf{j}{\bunj}{\ubj} - \numerdowj + \p \cdot (\bunj - \bundowj) = \cf{j}{\hbunj}{\ubj} - \numerdowj + \p \cdot (\hbunj - \bundowj),\]
for all agents $j$.
The claim follows.
\end{proof}

\begin{proof}[Proof of Claim~\ref{cl:regT}]
It suffices to show that $\payoffFn$ has convex values, range contained in $[\underline{M},\overline{M}]^J$, and a closed graph.

We first show that $\payoffs{\ubvec}$ is convex for all $\ubvec \in \bigtimes_{j \in J} [\umin{j},\umax{j}]$.
We use the notation of Claim~\ref{cl:hicksEqHelp} to prove this assertion.
Note that $\mathcal{P}$ is the set of solutions to a set of linear inequalities, and is hence convex.
Claim~\ref{cl:hicksEqHelp} implies that $\payoffs{\ubvec}$ is the set of values of a linear function on  $\mathcal{P}$---so it follows that $\payoffs{\ubvec}$ is convex as well.

We next show that $\payoffs{\ubvec} \subseteq [\underline{M},\overline{M}]^J$ holds for all $\ubvec \in \bigtimes_{j \in J} [\umin{j},\umax{j}].$
We again use the notation of Claim~\ref{cl:hicksEqHelp}.
Let $\ubvec \in \bigtimes_{j \in J} [\umin{j},\umax{j}]$ and $\payoffvec \in \payoffs{\ubvec}$ be arbitrary.
By Claim~\ref{cl:hicksEqHelp}, there exists $\p \in \mathcal{P}$ such that
\[\payoff{j} = \cf{j}{\bunj}{\ubj} - \numerdowj + \p \cdot (\bunj - \bundowj)\]
for all agents $j$.
Note that for all agents $j,$ we must have that
\[\payoff{j} \le \cf{j}{\bundowj}{\ubj} - \numerdowj \le \cf{j}{\bundowj}{\umax{j}} - \numerdowj \le \overline{M},\]
where the first inequality holds due to the definition of $\mathcal{P}$, the second inequality holds because $\cf{j}{\bundowj}{\cdot}$ is strictly increasing, and the third inequality holds due to the definition of $\overline{M}$.
Furthermore, as $\sum_{j \in J} \bunj = \tot = \sum_{j \in J} \bundowj$, we have that
\begin{align*}
\sum_{j \in J} \payoff{j} &= \sum_{j \in J} (\cf{j}{\bunj}{\ubj} - \numerdowj).
\end{align*}
It follows that
\begin{align*}
\payoff{j} &= \sum_{k \in J} (\cf{k}{\bunag{k}}{\ubag{k}} - \numerdowag{k}) - \sum_{k \in J \ssm \{j\}} \payoff{k}\\
&\ge \sum_{k \in J} (\cf{k}{\bunag{k}}{\umin{k}} - \numerdowag{k}) - \sum_{k \in J \ssm \{j\}} \payoff{k}\\
&\ge \sum_{k \in J} (\cf{k}{\bunag{k}}{\umin{k}} - \numerdowag{k}) - (|J|-1) \overline{M}\\
&\ge \underline{M}
\end{align*}
for all agents $j$, where the first inequality holds because $\cf{k}{\bunag{k}}{\cdot}$ is increasing for each agent $k$, the second inequality holds because $\payoff{k} \le \overline{M}$ for all agents $k$, and the third inequality holds due to the definition of $\underline{M}$.

\newcommand\ubvecseq[1]{\ubvec_{(#1)}}
\newcommand\ubjseq[1]{\ubj_{(#1)}}
\newcommand\payoffjseq[1]{\payoff{j}_{(#1)}}
\newcommand\payoffvecseq[1]{\payoffvec_{(#1)}}

Last, we show that $\payoffFn$ has a closed graph.
Our argument uses the following version of Farkas's Lemma.

\begin{fact}[Page 200 of \citealp{rockafellar1970convex}\footnote{Theorem 22.1 in \cite{rockafellar1970convex} states the case of Fact~\ref{fac:farkas} in which $L_1 = \emptyset$.  The version of Fact~\ref{fac:farkas} for $L_1 \not= \emptyset$ is left as an exercise on page 200 of \cite{rockafellar1970convex}.}]
\label{fac:farkas}
Let $L_1,L_2$ be disjoint, finite sets and, for each $\ell \in L_1 \cup L_2,$ let $\v^\ell \in \R^I$ be a vector and let $\alpha_\ell$ be a scalar.
There exist scalars $\lambda_\ell$ for $\ell \in L_1 \cup L_2$ with $\lambda_\ell \ge 0$ for $\ell \in L_2$ such that
\[\sum_{\ell \in L_1 \cup L_2} \lambda_\ell \v^\ell = \zero \quad \text{and} \sum_{\ell \in L_1 \cup L_2} \lambda_\ell \alpha_\ell < 0\]
if and only if there does not exist a vector $\p \in \R^I$ such $\v^\ell \cdot \p \le \alpha_\ell$ for all $\ell \in L_1 \cup L_2$ with equality for all $\ell \in L_1.$
\end{fact}

Consider a sequence $\ubvecseq{1},\ubvecseq{2},\ldots \in \bigtimes_{j \in J} [\umin{j},\umax{j}]$ of profiles of utility levels.
For each $m$, let $\payoffvecseq{m} \in \payoffs{\ubvecseq{m}}$.
Suppose that $\ubvecseq{m} \to \ubvec$ and $\payoffvecseq{m} \to \payoffvec$ as $m \to \infty$.
We need to show that $\payoffvec \in \payoffs{\ubvec}$.

As each set $\Feas{j}$ is finite and \andowalloc\ exists, by passing to a subsequence we can assume that there exists an allocation $(\bunj)_{j \in J} \in \bigtimes_{j \in J} \Feas{j}$ of goods with $\sum_{j \in J} \bunj = \tot$ that, for each $m,$ minimizes
\[\sum_{j \in J} \cf{j}{\hbunj}{\ubjseq{m}}\]
over all allocations $(\hbunj)_{j \in J} \in \bigtimes_{j \in J} \Feas{j}$ of goods with $\sum_{j \in J} \hbunj = \tot.$
By the continuity of $\cf{j}{\hbunj}{\ub}$ in $\ub$ for each agent $j$, the allocation $(\bunj)_{j \in J}$ minimizes \[\sum_{j \in J} \cf{j}{\hbunj}{\ubj}\]
over all allocations $(\hbunj)_{j \in J} \in \bigtimes_{j \in J} \Feas{j}$ of goods with $\sum_{j \in J} \hbunj = \tot.$

Suppose for sake of deriving a contradiction that $\payoffvec \notin \payoffs{\ubvec}$.
Let $L_1 = J$ and let $L_2 = \bigcup_{j \in J} \{j\} \times \Feas{j}.$
Define vectors $\v^\ell \in \mathbb{R}^I$ for $\ell \in L_1 \cup L_2$ by
\[
\v^\ell = \begin{cases}
\bunj - \bundowj & \text{ for } \ell = j \in L_1\\
\bunj - \bunpr & \text{ for } \ell = (j,\bunpr) \in L_2\\
\end{cases}
\]
and scalars $\alpha_\ell$ for $\ell \in L_1 \cup L_2$ by
\[
\alpha_\ell = \begin{cases}
\cf{j}{\bunj}{\ubj} - \numerdowj - \payoff{j} & \text{ for } \ell = j \in L_1\\
\cf{j}{\bunpr}{\ubj} - \cf{j}{\bunj}{\ubj} & \text{ for } \ell = (j,\bunpr) \in L_2.
\end{cases}
\]
By Claim~\ref{cl:hicksEqHelp}, there does not exist a price vector $\p$ such that $\v^\ell \cdot \p \le \alpha_\ell$ for all $\ell \in L_1 \cup L_2$ with equality for all $\ell \in L_1.$
The ``if'' direction of Fact~\ref{fac:farkas} therefore guarantees that there exist scalars $\lambda_\ell$ for $\ell \in L_1 \cup L_2$ with $\lambda_\ell \ge 0$ for all $\ell \in L_2$ such that
\[\sum_{\ell \in L_1 \cup L_2} \lambda_\ell \v^\ell = \zero \quad \text{and} \quad
\sum_{\ell \in L_1 \cup L_2} \lambda_\ell \alpha_\ell < 0.\]
By the definition of the scalars $\alpha_\ell,$ we have that
\[
\sum_{j \in J} \lambda_j \left(\cf{j}{\bunj}{\ubj} - \numerdowj - \payoff{j}\right) + \sum_{j \in J} \sum_{\bunpr \in \Feas{j}} \lambda_{j,\bunpr} \left(\cf{j}{\bunpr}{\ubj} - \cf{j}{\bunj}{\ubj}\right) < 0.
\]
Due the continuity of $\cf{j}{\hbunj}{\ub}$ in $\ub$ for each agent $j$ and because $\ubvecseq{m} \to \ubvec$ and $\payoffvecseq{m} \to \payoffvec$ as $m \to \infty$, there must exist $m$ such that
\[
\sum_{j \in J} \lambda_j \left(\cf{j}{\bunj}{\ubjseq{m}} - \numerdowj - \payoffjseq{m}\right) + \sum_{j \in J} \sum_{\bunpr \in \Feas{j}} \lambda_{j,\bunpr} \left(\cf{j}{\bunpr}{\ubjseq{m}} - \cf{j}{\bunj}{\ubjseq{m}}\right) < 0.
\]
Defining scalars $\alpha'_\ell$ for $\ell \in L_1 \cup L_2$ by
\begin{align*}
\alpha'_\ell &= \begin{cases}
\cf{j}{\bunj}{\ubjseq{m}} - \numerdowj - \payoffjseq{m} & \text{ for } \ell = j \in L_1\\
\cf{j}{\bunpr}{\ubjseq{m}} - \cf{j}{\bunj}{\ubjseq{m}} & \text{ for } \ell = (j,\bunpr) \in L_2,
\end{cases}
\end{align*}
we have that
\[\sum_{\ell \in L_1 \cup L_2} \lambda_\ell \v^\ell = \zero \quad \text{and that} \quad
\sum_{\ell \in L_1 \cup L_2} \lambda_\ell \alpha'_\ell < 0.\]
The ``only if'' implication of Fact~\ref{fac:farkas} therefore guarantees that there does not exist a price vector $\p$ such that $\v^\ell \cdot \p \le \alpha'_\ell$ for all $\ell \in L_1 \cup L_2$ with equality for all $\ell \in L_1.$
By Claim~\ref{cl:hicksEqHelp}, it follows that $\payoffvecseq{m} \notin \payoffs{\ubvecseq{m}}$---a contradiction.
Hence, we can conclude that $\payoffvec \in \payoffs{\ubvec}$---as desired.
\end{proof}

\subsubsection*{Completion of the Proof of the \ref{cond:marshall}\texorpdfstring{$\implies$}{ implies }\ref{cond:SWT} Implication in Theorem~\ref{thm:existDualExchangeSWT}.}

We first solve for a profile $\ubvec = (\ubj)_{j \in J}$ of utility levels such that $\zero \in \payoffs{\ubvec}.$

\begin{claim}
\label{cl:fpEq}
Under Statement~\ref{cond:hicks} in Theorem~\ref{thm:existDualExchangeSWT}, there exists a profile $\ubvec = (\ubj)_{j \in J}$ of utility levels such that $\zero \in \payoffs{\ubvec}$.
\end{claim}

To prove Claim~\ref{cl:fpEq}, we apply a topological fixed point argument.

\begin{proof}
Consider the compact, convex set
\[Z = [\underline{M},\overline{M}]^J \times \bigtimes_{j \in J} [\umin{j},\umax{j}].\]
As $\payoffs{\ubvec} \subseteq [\underline{M},\overline{M}]^J$ for all $\ubvec \in \bigtimes_{j \in J} [\umin{j},\umax{j}],$ we can define a correspondence $\Phi: Z \toto Z$ by
\[\Phi(\payoffvec,\ubvec) = \payoffs{\ubvec} \times \argmin_{\hubvec \in \bigtimes_{j \in J} [\umin{j},\umax{j}]} \left\{\sum_{j \in J} \payoff{j} \hubj\right\}.\]

Claim~\ref{cl:regT} guarantees that $T: \bigtimes_{j \in J} [\umin{j},\umax{j}] \toto \mathbb{R}^J$ is upper hemicontinuous and has compact, convex values.
Statement~\ref{cond:hicks} in Theorem~\ref{thm:existDualExchangeSWT} ensures that the correspondence $T$ has non-empty values.
Because $\bigtimes_{j \in J} [\umin{j},\umax{j}]$ is compact and convex, it follows that the correspondence $\Phi$ is upper hemicontinuous and has non-empty, compact, convex values as well.
Hence, Kakutani's Fixed Point Theorem guarantees that $\Phi$ has a fixed point $(\payoffvec,\ubvec)$.

By construction, we have that $\payoffvec \in \payoffs{\ubvec}$ and that
\begin{equation}
\label{eq:fpUtil}
\ubj \in \argmin_{\hubj \in [\umin{j},\umax{j}]} \payoff{j} \hubj
\end{equation}
for all agents $j$.
It suffices to prove that $\payoffvec = \zero$.

Let $\left(\p,(\bunnj)_{j \in J}\right)$ be a \ce\ in the \dualecon\ for the profile $(\ubj)_{j \in J}$ of utility levels with
\begin{equation}
\label{eq:fpPayoff}
\cf{j}{\bunj}{\ubj} - \numerdowj + \p \cdot (\bunj - \bundowj) = \payoff{j}
\end{equation}
for all agents $j$.
As $\ubj \ge \umin{j}$ and $\cf{j}{\bunj}{\cdot}$ is increasing for each agent $j$, it follows from Equation (\ref{eq:fpPayoff}) and the definition of $K^j$ that
\begin{align}
\payoff{j} &= \cf{j}{\bunj}{\ubj} - \numerdowj + \p \cdot (\bunj - \bundowj) \nonumber\\
&\ge \cf{j}{\bunj}{\umin{j}} - \numerdowj + \p \cdot (\bunj - \bundowj) \nonumber\\
&\ge \p \cdot (\bunj-\bundowj) - K^j \label{eq:tlb}
\end{align}
for all agents $j$.

Next, we claim that $\payoff{j} \le 0$ for all agents $j$.
If $\payoff{j} > 0,$ then Equation (\ref{eq:fpUtil}) would imply that $\ubj = \umin{j}$.
But as $\payoffvec \in \payoffs{\ubvec},$ it would follow that
\[\payoff{j} \le \cf{j}{\bundowj}{\umin{j}} - \numerdowj + \p \cdot (\bundowj - \bundowj) = \cf{j}{\bundowj}{\umin{j}} - \numerdowj = 0,\]
where the last equality holds due to the definitions of $\cfFn{j}$ and $\umin{j},$
so we must have that $\payoff{j} \le 0$ for all agents $j$.

As $(\bunj)_{j \in J}$ is the allocation of goods in a \ce, we have that $\sum_{j \in J} \bunj = \tot = \sum_{j \in J} \bundowj$ and hence that
\[\sum_{j \in J} \p \cdot (\bunj - \bundowj) = 0 \ge \sum_{j \in J} \payoff{j},\]
where the inequality holds because $\payoff{j} \le 0$ for all agents $j$.
It follows that for all agents $j,$ we have that
\[\payoff{j} - \p \cdot (\bunj - \bundowj) \le \sum_{k \in J \ssm \{j\}} (\p \cdot (\bunag{k} - \bundowag{k}) - \payoff{k}) \le \sum_{k \in J \ssm \{j\}} K^k \le \sum_{k \in J} K^k < K,\]
where the second inequality follows from Equation (\ref{eq:tlb}), the third inequality holds because $K^j \ge 0,$ and the fourth inequality holds due to the definition of $K$.
Hence, by Equation (\ref{eq:fpPayoff}), we have that
\[\cf{j}{\bunj}{\ubj} = \numerdowj + \payoff{j} - \p \cdot (\bunj - \bundowj) < \numerdowj + K\]
for all agents $j$.
Since utility is strictly increasing in the consumption of money, it follows that
\[\ubj = \util{j}{\cf{j}{\bunj}{\ubj},\bunj} < \util{j}{\numerdowj + K,\bunj} \le \umax{j},\]
where the equality holds due to the definition of $\cfFn{j}$ and the second inequality holds due to the definition of $\umax{j}$.
Equation (\ref{eq:fpUtil}) then implies that $\payoff{j} \ge 0$ for all agents $j$, so we must have that $\payoff{j} = 0$ for all agents $j$.
\end{proof}

By Claim~\ref{cl:fpEq}, there exists a profile $\ubvec = (\ubj)_{j \in J}$ of utility levels and a \ce\ $(\p,(\bunj)_{j \in j})$ in the corresponding \dualecon\ with
\begin{equation}
\label{eq:tFpFinal}
\numerdowj = \cf{j}{\bunj}{\ubj} + \p \cdot (\bunj - \bundowj)
\end{equation}
for all agents $j$.
Lemma~\ref{lem:dHvalH} implies that $\bunj \in \dH{j}{\p}{\ubj}$ for all agents $j$, and we have that $\util{j}{\numerdowj - \p \cdot (\bunj - \bundowj),\bunj} = \ubj$ for all agents $j$ by Equation (\ref{eq:tFpFinal}) and the definition of $\cfFn{j}$.
It follows from Fact~\ref{fac:dualDem} that $\bunj \in \dM{j}{\p}{\bunndowj}$ for all agents $j,$ so the price vector $\p$ and the allocation $(\bunj)_{j \in J}$ of goods comprise a \ce\ in the original economy for the endowment allocation $(\bunndowj)_{j \in J}$.


\section{Proof of Proposition~\ref{prop:ssubst}}
\label{app:grossToNet}

We actually prove a stronger statement.

\begin{claim}
\label{cl:marshallNetWeak}
Suppose that agent $j$ demands at most one unit of each good and let $\bundow \in \Feas{j}$.
A utility function $\utilFn{j}$ is a \nsubst\ utility function if for all money endowments $\numerdow > \feas{j},$ price vectors $\p,$ and $0 < \mu < \lambda$, whenever
\begin{enumerate}[label=(\roman*)]
    \item \label{cond:start} $\dM{j}{\p}{\bunndow} = \{\bun\}$,
    \item \label{cond:end} $\dM{j}{\p + \lambda \e{i}}{\bunndow} = \{\bunpr\}$,
    \item \label{cond:middle} $\{\bun,\bunpr\} \subseteq \dM{j}{\p + \mu \e{i}}{\bunndow}$, and
    \item \label{cond:ineq} $\bunprComp{i} < \bunComp{i}$,
\end{enumerate}
we have that $\bunprComp{k} \ge \bunComp{k}$ for all goods $k \not= i.$
\end{claim}

To complete the proof of the proposition from Claim~\ref{cl:marshallNetWeak}, we work in the setting of Claim~\ref{cl:marshallNetWeak}.
Note that, for the endowment $\bundow$ of goods, $\utilFn{j}$ is a \gsubst\ utility function when $\bunprComp{k} \ge \bunComp{k}$ holds for all goods $k \not= i$ under Conditions \ref{cond:start} and \ref{cond:end}.
This property clearly implies that $\bunprComp{k} \ge \bunComp{k}$ holds for all goods $k \not= i$ under Conditions \ref{cond:start}, \ref{cond:end}, \ref{cond:middle}, and \ref{cond:ineq}, and hence that $\utilFn{j}$ is \nsubst\ utility function by Claim~\ref{cl:marshallNetWeak}.
The proposition therefore follows from Claim~\ref{cl:marshallNetWeak}.

It remains to prove Claim~\ref{cl:marshallNetWeak}.
In the argument, we use the following characterization of \subst\ valuations.

\begin{fact}[Theorems 2.1 and 2.4 in \citealp{fujishige2003note}; Theorems 3.9 and 4.10(iii) in \citealp{shioura2015gross}]
\label{fac:subsDemCplx}
Suppose that agent $j$ demands at most one unit of each good.
A valuation $\valFn{j}$ is a \subst\ valuation if and only if for all price vectors $\p$ with $|\dQL{j}{\p}| = 2,$ writing $\dQL{j}{\p} = \{\bun,\bunpr\},$ the difference $\bunpr - \bun$ is a vector with at most one positive component and at most one negative component.
\end{fact}

\begin{proof}[Proof of Claim~\ref{cl:marshallNetWeak}]

We prove the contrapositive.
Suppose that $\utilFn{j}$ is not a \nsubst\ utility function.
We show that there exists a money endowment $\numerdow,$ a price vector $\p,$ price increments $0 < \mu < \lambda$, and goods $i \not= k$ such that Conditions \ref{cond:start}, \ref{cond:end}, \ref{cond:middle}, and \ref{cond:ineq} from the statement hold but $\bunprComp{k} < \bunComp{k}.$

By Remark~\ref{rem:netSubDual}, there exists a utility level $\ub$ such that $\valHDef{j}$ is not a \subst\ valuation.
Hence, by Lemma~\ref{lem:dHvalH} and 
the ``if" direction
of Fact~\ref{fac:subsDemCplx} for $\valFn{j} = \valHDef{j}$, there exists a price vector $\hp$ such that $|\dH{j}{\hp}{\ub}| = 2$, and writing $\dH{j}{\hp}{\ub} = \{\bun,\bunpr\},$ the difference $\bunpr - \bun$ has at least two positive components or at least two negative components.
Without loss of generality, we can assume that the difference $\bunpr - \bun$ has at least two negative components.
Suppose that $\bunprComp{i} < \bunComp{i}$ (so Condition~\ref{cond:ineq} holds) and that $\bunprComp{k} < \bunComp{k},$ where $i,k \in I$ are distinct goods.

Define a money endowment $\numerdow$ by \[\numerdow = \cf{j}{\bun}{\ub} + \hp \cdot (\bun - \bundow) = \cf{j}{\bunpr}{\ub} +  \hp \cdot (\bunpr - \bundow);\] Fact~\ref{fac:dualDem} implies that $\dM{j}{\hp}{\bunndow} = \{\bun,\bunpr\}$.
Let $\mu$ be such that
\[\dM{j}{\hp - \mu \e{i}}{\bunndow}, \dM{j}{\hp + \mu \e{i}}{\bunndow} \subseteq \{\bun,\bunpr\};\] such a $\mu$ exists due to the upper hemicontinuity of $\dMFn{j}$.
Let $\p = \hp - \mu \e{i},$ let $\lambda = 2 \mu,$ and let $\ppr = \p + \lambda \e{i} = \hp + \mu \e{i}$.

By construction, we have that $\{\bun,\bunpr\} \subseteq \dM{j}{\p + \mu \e{i}}{\bunndow} = \dM{j}{\hp}{\bunndow}$, so Condition \ref{cond:middle} holds.
It remains to show that $\dM{j}{\p}{\bunndow} = \{\bun\}$ and that $\dM{j}{\ppr}{\bunndow} = \{\bunpr\}$.
As $j$ demands at most one unit of each good, we must have that $\bunComp{i} = 1$ and that $\bunprComp{i} = 0.$
We divide into cases based on the value of $\bundowComp{i}$ to show that
\begin{equation}
\label{eq:subsIneqs}
\begin{aligned}
\util{j}{\numerdow - \p \cdot (\bun - \bundow),\bun} &> \util{j}{\numerdow - \p \cdot (\bunpr - \bundow),\bunpr}\\
\util{j}{\numerdow - \ppr \cdot (\bunpr - \bundow),\bunpr} &> \util{j}{\numerdow - \ppr \cdot (\bun - \bundow),\bun}.
\end{aligned}
\end{equation}
\begin{casework}
\item $\bundowComp{i} = 0$. In this case, we have that
\begin{align*}
\util{j}{\numerdow - \p \cdot (\bun - \bundow),\bun} &> \util{j}{\numerdow - \hp \cdot (\bun - \bundow),\bun}\\
&= \util{j}{\numerdow - \hp \cdot (\bunpr - \bundow),\bunpr}\\
&=\util{j}{\numerdow - \p \cdot (\bunpr - \bundow),\bunpr},
\end{align*}
where the inequality holds because $\pComp{i} < \hpComp{i}$ and $\bunComp{i} > \bundowComp{i}$, the first equality holds because $\{\bun,\bunpr\} \subseteq \dM{j}{\hp}{\bundow},$ and the second equality holds because $\bunprComp{i} = \bundowComp{i}.$
Similarly, we have that
\begin{align*}
\util{j}{\numerdow - \ppr \cdot (\bun - \bundow),\bun} &< \util{j}{\numerdow - \hp \cdot (\bun - \bundow),\bun}\\
&= \util{j}{\numerdow - \hp \cdot (\bunpr - \bundow),\bunpr}\\
&= \util{j}{\numerdow - \ppr \cdot (\bunpr - \bundow),\bunpr},
\end{align*}
where the inequality holds because $\pprComp{i} > \hpComp{i}$ and $\bunComp{i} > \bundowComp{i}$, the first equality holds because $\{\bun,\bunpr\} \subseteq \dM{j}{\hp}{\bundow},$ and the second equality holds because $\bunprComp{i} = \bundowComp{i}.$

\item $\bundowComp{i} = 1$. In this case, we have that
\begin{align*}
\util{j}{\numerdow - \p \cdot (\bunpr - \bundow),\bunpr}
&< \util{j}{\numerdow - \hp \cdot (\bunpr - \bundow),\bunpr}\\
&= \util{j}{\numerdow - \hp \cdot (\bun - \bundow),\bun}\\
&= \util{j}{\numerdow - \p \cdot (\bun - \bundow),\bun}
\end{align*}
where the inequality holds because $\pComp{i} < \hpComp{i}$ and $\bunprComp{i} < \bundowComp{i}$, the first equality holds because $\{\bun,\bunpr\} \subseteq \dM{j}{\hp}{\bundow},$ and the second equality holds because $\bunComp{i} = \bundowComp{i}.$
Similarly, we have that
\begin{align*}
\util{j}{\numerdow - \ppr \cdot (\bunpr - \bundow),\bunpr} &> \util{j}{\numerdow - \hp \cdot (\bunpr - \bundow),\bunpr}\\
&= \util{j}{\numerdow - \hp \cdot (\bun - \bundow),\bun}\\
&= \util{j}{\numerdow - \ppr \cdot (\bun - \bundow),\bun},
\end{align*}
where the inequality holds because $\pprComp{i} > \hpComp{i}$ and $\bunprComp{i} < \bundowComp{i}$, the first equality holds because $\{\bun,\bunpr\} \subseteq \dM{j}{\hp}{\bundow},$ and the second equality holds because $\bunComp{i} = \bundowComp{i}.$
\end{casework}
As $\bundow \in \Feas{j} \subseteq \{0,1\}^I,$ the cases exhaust all possibilities.
Hence, we have proven that Equation (\ref{eq:subsIneqs}) must hold.
As $\dM{j}{\p}{\bunndow},\dM{j}{\ppr}{\bunndow} \subseteq \{\bun,\bunpr\},$ we must have that $\dM{j}{\p}{\bunndow} = \{\bun\}$ and that $\dM{j}{\ppr}{\bunndow} = \{\bunpr\}$---so Conditions \ref{cond:start} and \ref{cond:end} hold, as desired.
\end{proof}

\singlespacing
\bibliographystyle{chicago}
\bibliography{bib}
\myspacing


\clearpage
\begin{center}
{\bf FOR ONLINE PUBLICATION}
\end{center}

\section{Proofs of Facts~\ref{fac:dualDem} and~\ref{fac:dualPrefs}}
\label{app:dualDemPrefs}

\subsection{Proof of Fact~\ref{fac:dualDem}}

We begin by proving two technical claims.

\begin{claim}
\label{cl:dualDemHelpMtoH}
Let $\bunndow \in \Feans{j}$ be an endowment and let $\ub$ be a utility level.
If
\begin{equation}
\label{eq:marshallForProofOfDual1}
\ub = \max_{\bunn \in \Feans{j} \mid \pall \cdot \bunn \le \pall \cdot \bunndow} \util{j}{\bunn},
\end{equation}
then we have that
\[\pall \cdot \bunndow = \min_{\bunn \in \Feans{j} \mid \util{j}{\bunn} \ge \ub} \pall \cdot \bunn\]
and that $\dM{j}{\p}{\bunndow} \subseteq \dH{j}{\p}{\ub}$.
\end{claim}
\begin{proof}
Letting $\bunpr \in \dM{j}{\p}{\bunndow}$ be arbitrary and $\numerpr = \numerdow - \p \cdot (\bunpr - \bundow),$ we have that $\util{j}{\bunnpr} = u$ and that $\pall \cdot \bunnpr \le \pall \cdot \bunndow$ by construction.
It follows that
\[\pall \cdot \bunndow \ge \min_{\bunn \in \Feans{j} \mid \util{j}{\bunn} \ge \ub} \pall \cdot \bunn.\]
Suppose for the sake of deriving a contradiction that
there exists $\bunndpr \in \Feans{j}$ with $\pall \cdot \bunndpr < \pall \cdot \bunndow$ and $\util{j}{\bunndpr} \ge \ub.$
Then, we have that $\numerdpr < \numerdow + \p \cdot (\bundpr - \bundow)$; write $\numertpr = \numerdow + \p \cdot (\bundpr - \bundow),$ so $\numertpr > \numerdpr.$
Since $\utilFn{j}$ is strictly increasing in consumption of money, it follows that $\util{j}{\numertpr,\bundpr} > u$---contradicting Equation (\ref{eq:marshallForProofOfDual1}) as $\numertpr + \p \cdot \bundpr = \pall \cdot \bunndow.$
Hence, we can conclude that
\[\pall \cdot \bunndow = \min_{\bunn \in \Feans{j} \mid \util{j}{\bunn} \ge \ub} \pall \cdot \bunn.\]

Since $\util{j}{\bunnpr} = \ub$ and $\pall \cdot \bunnpr = \pall \cdot \bunndow,$ it follows that $\bunpr \in \dH{j}{\p}{\ub}$.
Since $\bunpr \in \dM{j}{\p}{\bunndow}$ was arbitrary, we can conclude that $\dM{j}{\p}{\bunndow} \subseteq \dH{j}{\p}{\ub}$.
\end{proof}

\begin{claim}
\label{cl:dualDemHelpHtoM}
Let $\bunndow \in \Feans{j}$ be an endowment and let $\ub$ be a utility level.
If
\begin{equation}
\label{eq:hicksForProofOfDual1}
\pall \cdot \bunndow = \min_{\bunn \in \Feans{j} \mid \util{j}{\bunn} \ge \ub} \pall \cdot \bunn,
\end{equation}
then we have that
\[\ub = \max_{\bunn \in \Feans{j} \mid \pall \cdot \bunn \le \pall \cdot \bunndow} \util{j}{\bunn}\]
and that $\dH{j}{\p}{\ub} \subseteq \dM{j}{\p}{\bunndow}$.
\end{claim}
\begin{proof}
Let $\bunpr \in \dH{j}{\p}{\ub}$ be arbitrary and $\numerpr = \cf{j}{\bunpr}{\ub}$.
We have that $\util{j}{\bunnpr} \ge \ub$ and that $\pall \cdot \bunnpr = \pall \cdot \bunndow$ by construction.
It follows that
\[\ub \le \max_{\bunn \in \Feans{j} \mid \pall \cdot \bunn \le \pall \cdot \bunndow} \util{j}{\bunn}.\]

We next show that 
\[\ub = \max_{\bunn \in \Feans{j} \mid \pall \cdot \bunn \le \pall \cdot \bunndow} \util{j}{\bunn}.\]
Suppose for sake of deriving a contradiction that
there exists $\bunndpr \in \Feans{j}$ with $\pall \cdot \bunndpr \le \pall \cdot \bunndow$ and $\util{j}{\bunndpr} > \ub.$
By definition of $\cfFn{j}$, we know that $\numerdpr > \cf{j}{\bundpr}{\ub}$.
Letting $\numertpr = \cf{j}{\bundpr}{\ub}$, we have that $\numertpr + \p \cdot \bundpr < \pall \cdot \bunndow$, which contradicts Equation (\ref{eq:hicksForProofOfDual1}) as $\util{j}{\numertpr,\bundpr} = \ub.$
Hence, we can conclude that
\[\ub = \max_{\bunn \in \Feans{j} \mid \pall \cdot \bunn \le \pall \cdot \bunndow} \util{j}{\bunn}.\]

Since $\util{j}{\bunpr} = u$ and $\pall \cdot \bunpr = \pall \cdot \bunndow,$ it follows that $\bunpr \in \dM{j}{\p}{\bunndow}$.
Since $\bunpr \in \dH{j}{\p}{\ub}$ was arbitrary, we can conclude that $\dH{j}{\p}{\ub} \subseteq \dM{j}{\p}{\bunndow}$.
\end{proof}

Let $\bunndow \in \Feans{j}$ be an endowment and let $\ub$ be a utility level.
By Claims~\ref{cl:dualDemHelpMtoH} and~\ref{cl:dualDemHelpHtoM}, Conditions (\ref{eq:marshallForProofOfDual1}) and (\ref{eq:hicksForProofOfDual1}), are equivalent, and under these equivalent conditions, we have that $\dM{j}{\p}{\bunndow} \subseteq \dH{j}{\p}{\ub}$ and that $\dH{j}{\p}{\ub} \subseteq \dM{j}{\p}{\bunndow}$.
Hence, we must have that $\dM{j}{\p}{\bunndow} = \dH{j}{\p}{\ub}$ under the equivalent Conditions (\ref{eq:marshallForProofOfDual1}) and (\ref{eq:hicksForProofOfDual1})---as desired.

\subsection{Proof of Fact~\ref{fac:dualPrefs}}

We prove the ``if'' and ``only if'' directions separately.

\subsubsection*{Proof of the ``If" Direction.}

We define a utility function $\utilFn{j}$ implicitly by
\[\util{j}{\bunn} = F(\bun,\cdot)^{-1}(-\numer),\]
which is well-defined, continuous, and strictly increasing in $\numer$ by the Inverse Function Theorem because $F(\bun,\cdot)$ is continuous, strictly decreasing, and satisfies Condition (\ref{eq:vlimits}).
Condition (\ref{eq:ulimits}) holds because $F$ is defined over the entirety of $\Feas{j} \times \feasUtil{j}.$

\subsubsection*{Proof of the ``Only If" Direction.}

We define $F:\Feas{j} \times (\minu{j},\maxu{j}) \to (-\infty,-\feas{j})$ implicitly by
\[F(\bun,\ub) = -\util{j}{\cdot,\bun}^{-1}(\numer),\]
which is well-defined, continuous, and strictly decreasing in $\ub$ by the Inverse Function Theorem because $\util{j}{\cdot,\bun}$ is continuous, strictly increasing, and satisfies Condition (\ref{eq:ulimits}).
Condition (\ref{eq:vlimits}) holds because $\utilFn{j}$ is defined over the entirety of $\Feans{j}$.

\section{Proofs of the Maximal Domain and Necessity Results for Settings with Transferable Utility}
\label{app:maxDomain}

In this appendix, we supply proofs of Facts~\ref{fac:subMaxDomain} and~\ref{fac:unimodMaxDomain}, as well as the ``only if" direction of Fact~\ref{fac:unimod}.  Utility is transferable throughout this appendix.

We use the concept of a pseudo-equilibrium price vector.

\begin{definition}[\citealp{MiSt:09}]
Suppose that utility is transferable.  A \emph{pseudo-equilibrium price vector} is a price vector $\p$ such that
\[\tot\in \conv\left(\sum_{j \in J} D^j(\p)\right).\]
\end{definition}

There is a connection between pseudo-equilibrium price vectors, \ces, and the existence problem.

\begin{fact}[Theorem 18 in \citealp{MiSt:09}; Lemma 2.19 in \citealp{BaKl:19}]
\label{fac:pseudoEquil}
If utility is transferable and the total endowment is such that a competitive equilibrium exists, then, for each pseudo-equilibrium price vector $\p$, there exists an allocation $(\bunj)_{j \in J}$ such that $\p$ and  $(\bunj)_{j \in J}$ comprise a \ce.
\end{fact}

The nonexistence of \ces\ may therefore be demonstrated by using the contrapositive of Fact~\ref{fac:pseudoEquil}.

Our arguments use valuations that are \emph{linear on their domain}.  That is, let $\mathbf{t}_I\in\R^I$, let $X^j_I\subseteq\Z^n$ be finite, and let $V^j=V^{j,\mathbf{t}_I}\mathbf:X^j_I\rightarrow\R$ be given by $V^{j,\mathbf{t}_I}(\mathbf{x}_I)\coloneq\mathbf{t}_I\cdot\mathbf{x}_I$ for all $\mathbf{x}_I\in X^j_I$.  Recalling Equation (\ref{eqn:quasilin}) for demand sets in the quasilinear case, we observe that, for each $\mathbf{s}_I\in\R^I$, we have that
\begin{equation}\label{eqn:faces}
D^j(\mathbf{t}_I-\mathbf{s}_I)=\argmax_{\mathbf{x}_I\in X^j_I}\left(\mathbf{t}_I\cdot\mathbf{x}_I-(\mathbf{t}_I-\mathbf{s}_I)\cdot\mathbf{x}_I\right)=\argmax_{\mathbf{x}_I\in X^j_I}\mathbf{s}_I\cdot\mathbf{x}_I.
\end{equation}

\begin{lemma}\label{lem:linConv}
If $\conv(X^j_I)\cap\Z^I=X^j_I$, then $V^{j,\mathbf{t}_I}$ is concave for all $\mathbf{t}_I\in\R^I$.
\end{lemma}
\begin{proof}
Observe by Equation (\ref{eqn:faces}) that $D^j(\mathbf{t}_I)=\argmax_{\mathbf{x}_I\in X^j_I}\mathbf{0}\cdot\mathbf{x}_I=X^j_I$.  So, if $\mathbf{x}_I\in\conv(X^j_I)\cap\Z^I=X^j_I$ then $\mathbf{x}_I\in D^j(\mathbf{t}_I)$.  By Definition \ref{def:concave}, we know $V^{j,\mathbf{t}_I}$ is concave.
\end{proof}

We will also make use of an alternative characterization of concavity.
\begin{fact}[Lemma 2.11 in \citealp{BaKl:19}]
\label{fac:conc}
A valuation $\valFn{j}$ is concave if and only if $\conv\left(\dQL{j}{\p}\right) \cap \Z^I = \dQL{j}{\p}$ for all price vectors $\p$.
\end{fact}

\subsection{Additional Facts regarding Unimodularity and Demand Types}

The following results are especially useful in the proof of the ``only if'' direction of Fact~\ref{fac:unimod}, and the proof of Fact~\ref{fac:unimodMaxDomain}.

We seek to construct pseudo-equilibrium price vectors (the total endowment is in the convex hull of aggregate demand) that are not \ce\ price vectors (the total endowment is not demanded on aggregate).   Failure of unimodularity allows such constructions because of the following property.

\begin{fact}[See, e.g., Fact 4.9 in \cite{BaKl:19}]
\label{fac:unimodSet}
A demand type vector set $\mathcal{D}$ is unimodular if and only if there is no linearly independent subset $\{\mathbf{d}^1,\ldots,\mathbf{d}^r\}$ of $\mathcal{D}$ such that there exists $\mathbf{z}=\sum_{\ell=1}^r \alpha_\ell \dvec^\ell\in\Z^I$ with $\alpha_\ell\in(0,1)$ for $\ell=1,\ldots,r$.
\end{fact}

To see the connection to Fact \ref{fac:pseudoEquil} and to existence of competitive equilibrium,  suppose that $\{\mathbf{d}^1,\ldots,\mathbf{d}^r\}$ and $\mathbf{z}$ are as in Fact \ref{fac:unimodSet}.  If $\tot=\mathbf{z}$ and if $D^j_M(\p,\bunndowj)=\{\mathbf{0},\mathbf{d}^j\}$ for $j=1,\dots,r$, then $\p$ is a pseudo-equilibrium price vector but there is no competitive equilibrium at $\p$.

\cite{BaKl:19} generalized Fact~\ref{fac:subsDemCplx} to the general case of transferable utility.

\begin{fact}[Proposition 2.20 in \citealp{BaKl:19}]
\label{fac:demTypeCplx}
Let $\valFn{j}$ be a valuation of demand type $\mathcal{D}$.  For any price $\ppr$, if $\conv(\dQL{j}{\ppr})$ has an edge $E$, then the difference between the extreme points of $E$ is proportional to a demand type vector, and there exists a price $\p$ such that $\conv(\dQL{j}{\p})=E$.

Moreover if $\dvec$ is in the minimal \dtvs\ $\mathcal{D}$, such that $\valFn{j}$ is of demand type $\mathcal{D}$, then there exists a price vector $\p$ such that $\conv(\dQL{j}{\p})$ is a line segment, the difference between whose endpoints is proportional to $\dvec.$  
\end{fact}

We also demonstrate now the following useful corollary of Fact \ref{fac:demTypeCplx}.

\begin{corollary}\label{cor:edges}
Let $V^j=V^{j,\mathbf{t}_I}$ for some $\mathbf{t}_I\in\R^I$, and let $\mathcal{D}$ be the minimal demand type vector set such that $V^j$ is of demand type $\mathcal{D}$.  Then $\mathcal{D}$ consists of the primitive integer vectors in the directions of the edges of the polytope $\conv(X^j_I)$.
\end{corollary}
\begin{proof}
Observe that $D^j(\mathbf{t}_I)=X^j_I$ and so, by Fact \ref{fac:demTypeCplx}, each edge of $\conv(X^j_I)$ is proportional to a vector in $\mathcal{D}$.  Conversely, if $\mathbf{d}\in\mathcal{D}$ then, by Fact \ref{fac:demTypeCplx}, there exists a price $\p$ such that $\conv(D^j(\p))$ is a line segment, the difference between whose endpoints is proportional to $\mathbf{d}$.  But writing $\mathbf{s}_I=\mathbf{t}_I-\p$, we see from Equation (\ref{eqn:faces}) that $D^j(\p)=\argmax_{x_I\in X^j_I} \mathbf{s}_I\cdot\mathbf{x}_I$ which tells us (cf. e.g.\ \citealt{Gruenbaum1967}, Section 2.4) that $E$ is an edge of $\conv(X^j_I)$.
\end{proof}

Our proofs of Facts \ref{fac:subMaxDomain} and~\ref{fac:unimodMaxDomain}, and the ``only if'' direction of Fact~\ref{fac:unimod}, now follow the same structure.
Within each argument, we address a demand type which is not unimodular.  Observe by Fact \ref{fac:unimodSet} that when unimodularity fails for a set of vectors $\mathcal{D}$, then there exist polytopes, with integer vertices and whose edge directions are in $\mathcal{D}$, that contain a non-vertex integer vector, $\mathbf{z}$.  We use Corollary~\ref{cor:edges} construct valuations of the appropriate demand type such that, at some price $\p$, the convex hull of the aggregate demand set is a polytope with these properties; and such that there exists a feasible endowment allocation is the total endowment is the non-vertex integer vector $\mathbf{z}$.  Thus $\p$ is a pseudo-equilibrium price.  Moreover, we design our individual valuations so that $\p$ is not a competitive equilibrium.  The contrapositive of Fact \ref{fac:pseudoEquil} can then be applied to show the non-existence of competitive equilibrium.

\subsection{Proof of Fact~\ref{fac:subMaxDomain}}

By Fact~\ref{fac:subsDemCplx}, there exists a price vector $\p$ such that $\dQL{j}{\p} = \{\bunpr,\bunpr+\norm\}$, where $\norm$ has at least two positive components or at least two negative components.  Identify $I$ with $\{1,\ldots,|I|\}$ and without loss of generality assume that $g_1,g_2<0$. Because agent $j$ demands at most one unit of each good, we know that 
$\bunpr,\bunpr+\norm\in\{0,1\}^{|I|}$ and so $\norm\in \{-1,0,1\}^{|I|}$.  We conclude both that $g_1=g_2=-1$ and that $x'_1=x'_2=1$.

Let $k \in J \ssm \{j\}$ be arbitrary.  For agents $j' \in J \ssm \{j,k\},$ let $\Feas{j'} = \{\zero\}$, let $\valFn{j'}$ be arbitrary, and let $\bundowag{j'} = 0$.

Let $X^k_I\coloneq\{\rgivend\mathbf{x}_I\in\{0,1\}^{|I|}\rgiv x_1+x_2\leq 1\}$ and let $\mathbf{t}_I\coloneq\p-\mathbf{e}^1-\mathbf{e}^2$.  Let $V^k = V^{k,{\bf t}_I},$
which is a \subst\ valuation by Example 11 and Corollary D.2, because each edge of $\conv(X^k_I)$ is proportional to either ${\bf e}^1 - {\bf e}^2$ or to ${\bf e}^\ell$ for some $\ell \in I$; or, alternatively, by Theorem 4 in Hatfield et al. (2019) because, like \citeposs{shapley1962complements} assignment valuations, it is the supremal convolution of $|I|-1$ unit-demand valuations.  

Set $\mathbf{w}^j_I\coloneq\bunpr\in X^j_I$ and set $\mathbf{w}^k_I\coloneq\tot-\bunpr$.  Since $\bunpr\in\{0,1\}^{|I|}$ it follows that $\mathbf{w}^k_I\in\{0,1\}^{|I|}$, and moreover since $x'_1=x'_2=1$ we know $w^k_1=w^k_2=0$; thus $\mathbf{w}^k_I\in X^k_I$.   Now $(\mathbf{w}^{j'}_I)_{j'\in J}$ is clearly an endowment allocation.

By Equation (\ref{eqn:faces}) we know that 
\begin{equation*}
D^k(\mathbf{p}_I)=\argmax_{\mathbf{x}_I\in X^k_I}(\mathbf{e}^1+\mathbf{e}^2)\cdot\mathbf{x}_I=\{\mathbf{x}_I\in\{0,1\}^{|I|}|x_1+x_2=1\}.
\end{equation*}
Observing that $\mathbf{e}^2\in D^k(\mathbf{p}_I)$ and considering the vectors from $\mathbf{e}^2$ to other elements of the demand set, we can write $D^k(\mathbf{p}_I)$ as
\[
D^k(\p) = \e{2} + \left\{\rgivend \alpha_2 (\e{1} - \e{2}) + \sum_{\ell = 3}^{|I|} \alpha_\ell \e{\ell} \rgiv \alpha_\ell \in \{0,1\} \text{ for } 2 \le \ell \le |I|\right\}.
\]
Combining this with agent $j$, and recalling other agents' demand sets are identically zero, we conclude that
\[
\sum_{j'\in J} D^{j'}(\mathbf{p}_I)=\bunpr+\mathbf{e}^2+\left\{\rgivend\alpha_1\norm+\alpha_2(\mathbf{e}^1-\mathbf{e}^2)+\sum_{\ell=3}^{|I|}\alpha_\ell\mathbf{e}^\ell\rgiv\alpha_\ell\in\{0,1\}\text{ for }1\leq \ell\leq |I|\right\}.
\]
The convex hull of this set can be expressed very similarly, but the weights $\alpha_\ell$ are allowed to lie in $[0,1]$.

Since $\bunpr, \bunpr+\mathbf{g}\in\{0,1\}^{|I|},$ we have that if $g_i=1$ (resp.~$g_i = -1$), then $x'_i=0$ (resp.~$x'_i = 1$).
Taking
\[\alpha_\ell = \begin{cases}
\dfrac{|\normComp{\ell}|}{2} & \text{ if } \normComp{\ell} \not= 0\\
1 - \bunprComp{\ell} & \text{if } \normComp{\ell} = 0
\end{cases}\]
for $1 \le \ell \le |I|$, we have that
\begin{align*}
\bunprComp{i} + \alpha_1 \normComp{i} + \alpha_i &= 1 - \frac{1}{2} + \frac{1}{2} = 1  & \text{ for all } i \in I \text{ with } \normComp{i} = -1\\
\bunprComp{i} + \alpha_1 \normComp{i} + \alpha_i &= \bunprComp{i} + 0 + (1 - \bunprComp{i}) = 1 & \text{ for all } i \in I \text{ with } \normComp{i} = 0\\
\bunprComp{i} + \alpha_1 \normComp{i} + \alpha_i &= 0 + \frac{1}{2} + \frac{1}{2} = 1  & \text{ for all } i \in I \text{ with } \normComp{i} = 1.
\end{align*}
As $\bunprComp{1} = \bunprComp{2} = 1$ and $\normComp{1} = \normComp{2} = -1,$ it follows that
\[\bunpr + \e{2} + \alpha_1\norm+\alpha_2(\mathbf{e}^1-\mathbf{e}^2)+\sum_{\ell=3}^{|I|}\alpha_\ell\mathbf{e}^\ell = \tot.\]

As $\alpha_\ell \in [0,1]$ for all $1 \le \ell \le |I|$, we therefore have that $\tot \in \conv\left(\sum_{j'\in J} D^{j'}(\mathbf{p}_I)\right)$, so $\p$ is a pseudo-equilibrium price vector.
But as $\alpha_1 \in (0,1)$ and the vectors $\mathbf{g},\mathbf{e}^1-\mathbf{e}^2,\mathbf{e}^3,\ldots,\mathbf{e}^{|I|}$ are linearly independent, we have that $\tot \notin \sum_{j'\in J} D^{j'}(\mathbf{p}_I),$ so there is no \ce\ at $\p$.\footnote{The existence of an integer vector that is in $\conv\left(\sum_{j'\in J} D^{j'}(\mathbf{p}_I)\right)$ but not $\sum_{j' \in J} D^{j'}(\mathbf{p}_I)$ follows from Fact~\ref{fac:unimodSet} as the vectors $\mathbf{g},\mathbf{e}^1-\mathbf{e}^2,\mathbf{e}^3,\ldots,\mathbf{e}^{|I|}$ do not comprise a unimodular set.}
Therefore, by the contrapositive of Fact~\ref{fac:pseudoEquil}, no \ce\ can exist.

\subsection{Proof of the ``Only If'' Direction of Fact~\ref{fac:unimod}}

Let $\mathcal{D}$ be a \dtvs\ that is not unimodular.
We need to show that there exists a finite set $J$ of agents with concave valuations of demand type $\mathcal{D}$, as well as a total endowment, for which there exists \andowalloc\ but no \ce.  We will use $J=\{j,k\}$.

Let $\linset = \{\dvec^1,\ldots,\dvec^n\} \subseteq \mathcal{D}$ be a minimal non-unimodular subset.
By construction, $\linset$ is linearly independent, and $\{\dvec^1,\ldots,\dvec^{n-1}\}$ is unimodular.  Let
\[
\mathcal{P} = \left\{\rgivend\sum_{\ell=1}^n \alpha_\ell \dvec^\ell \rgiv 0 \le \alpha_\ell \le 1\text{ for }\ell=1,\ldots,n\right\}
\]
denote the parallelepiped spanned by $\linset$.  By Fact \ref{fac:unimodSet}, there exists $\mathbf{z} = \sum_{\ell=1}^n \beta_\ell \dvec^\ell \in \mathcal{P} \cap \mathbb{Z}^I$ with $\beta_\ell \in (0,1)$ for all $\ell=1,\ldots,n$.

Let $X^j_I\coloneq \mathcal{P}\cap\Z^n$ and let $V^j\coloneq V^{j,\mathbf{0}}$ be the linear valuation which is identically zero on its domain.  Recall Equation (\ref{eqn:faces}): we know $D^j(\mathbf{0})=X^j_I$.  Observe that $\mathbf{z}\in X^j_I$.  Clearly $\conv(X^j_I)\cap\Z^I=X^j_I$ and so $V^j$ is concave by Lemma \ref{lem:linConv}.

Let $\mathbf{s}_I$ satisfy $\mathbf{s}_I\cdot \dvec^\ell= 0$ for $\ell=1,\ldots,n-1$ and $\mathbf{s}_I\cdot\dvec^n>0$. (Such an $\mathbf{s}_I$ exists as $\linset$ is linearly independent.)  Then, for $\mathbf{x}_I=\sum_{\ell=1}^n \alpha_\ell \dvec^\ell\in X^j_I$, we have \[\mathbf{s}_I\cdot\mathbf{x}_I=\sum_{\ell=1}^n \alpha_\ell \mathbf{s}_I\cdot\dvec^\ell=\alpha_n \mathbf{s}_I\cdot\dvec^n.\]  We assumed that $\mathbf{s}_I\cdot\dvec^n>0$, so $\mathbf{s}_I\cdot\mathbf{x}_I$ is minimized when $\alpha_n=0$; equivalently $-\mathbf{s}_I\cdot\mathbf{x}_I$ is 
maximized when $\alpha_n=0$.  So, by Equation (\ref{eqn:faces}), we know that 
\begin{equation*}
D^j(\mathbf{s}_I)=\argmax_{\mathbf{x}_I\in X^j_I}-\mathbf{s}_I\cdot\mathbf{x}_I=\left\{\rgivend\sum_{\ell=1}^{n-1} \alpha_\ell \dvec^\ell \rgiv 0 \le \alpha_\ell \le 1 \text{ for }\ell=1,\ldots,n-1\right\}\cap\Z^I.
\end{equation*}

Now set $X^k_I\coloneq\{\mathbf{0},\dvec^n\}$ and let $V^k\coloneq V^{k,\mathbf{s}_I}$.  By Equation (\ref{eqn:faces}) again, we know that $D^k(\mathbf{s}_I)=X^k_I$.  As $\mathbf{d}^n\in\mathcal{D}$, which is a \dtvs, we know that $\mathbf{d}^n$ is a primitive integer vector, from which it follows that $\conv(X^k_I)\cap\Z^I=X^k_I$.  Thus, by Lemma \ref{lem:linConv}, we know that $V^k$ is concave.

Observe that 
\[
D^j(\mathbf{s}_I)+D^k(\mathbf{s}_I)=\left\{\rgivend\sum_{\ell=1}^n \alpha_\ell \dvec^\ell \rgiv 0 \le \alpha_\ell \le 1 \text{ for }\ell=1,\ldots,n-1 \text{ and }\alpha_n\in\{0,1\}\right\} \cap \Z^I.
\]
So $\conv(D^j(\mathbf{s}_I)+D^k(\mathbf{s}_I))=\mathcal{P}$.

Let the total endowment $\tot$ be $\mathbf{z}$.  Set $\mathbf{w}^j_I\coloneq\mathbf{z}\in X^j_I$, and set $\mathbf{w}^k_I\coloneq\mathbf{0}\in X^k_I$.  This is clearly an endowment allocation.  Since $\tot\in\mathcal{P}=\conv(D^j(\mathbf{s}_I)+D^k(\mathbf{s}_I))$, we see  $\mathbf{s}_I$ is a pseudo-equilibrium price vector.  But, since $\linset$ is linearly independent and since $0<\beta_n<1$, we know $\tot\notin D^j(\mathbf{s}_I)+D^k(\mathbf{s}_I)$, so there is no competitive equilibrium at $\mathbf{s}_I$.  It follows, by the contrapositive of Fact~\ref{fac:pseudoEquil}, that no \ce\ can exist.

\subsection{Proof of Fact~\ref{fac:unimodMaxDomain}}
We will use the following claim.

\begin{claim}\label{cla:maximal}
If $\mathcal{D}$ is a maximal unimodular \dtvs\ then $\mathcal{D}$ spans $\mathbb{R}^I$.
\end{claim}
\begin{proof}
Let $\linset \subseteq \mathcal{D}$ be a maximal, linearly independent set.
As $\mathcal{D}$ is unimodular, there exists a set $T$ of \intvecs\ such that $\linset \cap T = \emptyset$ and $\linset \cup T$ is a basis of $\mathbb{R}^I$ with determinant $\pm 1.$
We claim that $\mathcal{D}_0 = \mathcal{D} \cup T \cup -T$ is unimodular.
To see why, let $L' \subseteq \mathcal{D} \cup T\cup -T$ be a maximal linearly independent set.
As $\mathcal{D}_0$ spans $\mathbb{R}^I$ by construction, $L'$ must span $\mathbb{R}^I$.
Due to the maximality of $\linset$, we must have that $|L' \cap (T \cup -T)| = |T|.$
It follows that $L' \cap \mathcal{D}$ is a basis for the span of $\mathcal{D}$.
As $\mathcal{D}$ is unimodular, $L' \cap \mathcal{D}$ must be the image of $\linset$ under a unimodular change of basis of the span of $\mathcal{D}$.
It follows that $L'$ is a basis for $\mathbb{R}^I$ with determinant $\pm 1$---so $\mathcal{D}_0$ is unimodular.
Due to the maximality of $\mathcal{D},$ we must have that $T = \emptyset$, and hence $\mathcal{D}$ must span $\mathbb{R}^I$.
As $\mathcal{D}$ is unimodular, it follows that $\mathcal{D}$ must integrally span $\mathbb{Z}^I$.
\end{proof}

We next divide into cases based on whether $\valFn{j}$ is non-concave and of demand type $\mathcal{D}$, or not of demand type $\mathcal{D}$, to construct concave valuations $\valFn{k}$ of demand type $\mathcal{D}$ for agents $k \not= j$ and a total endowment for which no \ce\ exists.

\begin{casework}
\item $\valFn{j}$ is not concave but is of demand type $\mathcal{D}$.  By Fact~\ref{fac:conc}, there exists a price vector $\p$ such that $\dQL{j}{\p} \not= \conv\left(\dQL{j}{\p}\right) \cap \Z^I.$
Let $\bunpr \in \dQL{j}{\p}$ be an extreme point of $\conv\left(\dQL{j}{\p}\right)$, so there exists $\mathbf{s}_I\in\R^n$ satisfying 
\begin{equation}\label{eqn:xprime}
\{\bunpr\}=\argmax_{\mathbf{y}_I\in \dQL{j}{\p}}\mathbf{s}_I\cdot\mathbf{y}_I.
\end{equation}
Let 
$\bundpr \in \left(\conv\left(\dQL{j}{\p}\right) \cap \Z^I\right) \ssm \dQL{j}{\p}$
be arbitrary.  

Let $k \in J \ssm \{j\}$ be arbitrary.  Let $X^k_I=(\conv(D^j(\p))\cap\Z^I)+\{-\bunpr\}$.
Since $V^j$ is of demand type $\mathcal{D}$, it follows by Fact \ref{fac:demTypeCplx} that every edge of $\conv(D^j(\p))$ is a multiple of a vector in $\mathcal{D}$, and so the same is true of $\conv(X^k_I)$.   Moreover, by definition of $X^k_I$ it is clear that $\conv(X^k_I)\cap\Z^I=X^k_I$.

Fix $\mathbf{t}_I\coloneq\p+\mathbf{s}_I$ and let $V^k\coloneq V^{\mathbf{t}_I,k}$, which is concave by Lemma \ref{lem:linConv} and of demand type $\mathcal{D}$ by Corollary \ref{cor:edges}.   By Equation
(\ref{eqn:faces}) we know $D^k(\p)=\argmax_{\mathbf{x}_I\in X^k_I}\mathbf{s}_I\cdot\mathbf{x}_I$, and so by Equation (\ref{eqn:xprime}) and the definition of $X^k_I$, it follows that $D^k(\p)=\{\bunpr-\bunpr\}=\{\mathbf{0}\}$.

Let the total endowment $\tot$ be $\mathbf{x}''_I$, let $\bundowj \coloneq \bunpr\in X^j_I$, and let $\mathbf{w}^k_I\coloneq\bundpr-\bunpr\in X^k_I$. For agents $j' \in J \ssm \{j,k\}$, let $\Feas{j'} = \{\zero\}$, let $\valFn{j'}$ be arbitrary, and let $\bundowag{j'} = \zero.$  Thus $(\mathbf{w}^{j'}_I)_{j'\in J}$ is an endowment allocation.  Moreover, 
\[
\sum_{j'\in J}D^{j'}(\p)=D^j(\p).
\]
Thus $\tot=\bundpr\in \conv\left(\sum_{j'\in J}D^{j'}(\p)\right)$ and so $\p$ is a pseudo-equilibrium price vector.  But $\tot=\bundpr\notin \sum_{j'\in J}D^{j'}(\p)=D^j(\p)$ by definition of $\bundpr$, and so there is no \ce\ at $\p$.  Therefore, by the contrapositive of Fact~\ref{fac:pseudoEquil}, no \ce\ can exist.

\item $V^j$ is not of demand type $\mathcal{D}$.  By Fact \ref{fac:demTypeCplx} there exists a primitive integer vector $\norm\notin\mathcal{D}$ and a price vector $\p\in\R^n$ such that $D^j(\p)\subseteq\{\bunpr+\alpha\norm\ | \ \alpha=0,\ldots,r\}$ where $r\geq 1$ and $\bunpr,\bunpr+r\norm\in D^j(\p)$.   

As $\mathcal{D}$ is not strictly contained in any unimodular \dtvs, and as $\norm\notin\mathcal{D}$, the set $\mathcal{D}\cup\{\norm\}$ is not unimodular.  Let $\{\mathbf{d}^1,\ldots,\mathbf{d}^m,\norm\}$ be a minimal non-unimodular subset of $\mathcal{D}\cup\{\norm\}$.  Thus the set $\{\mathbf{d}^1,\ldots,\mathbf{d}^m,\norm\}$ is linearly independent and, by Fact \ref{fac:unimodSet}, there exists 
\begin{equation}\label{eqn:zinP}
\mathbf{z}=\beta_0\norm+\sum_{\ell=1}^m\beta_\ell \mathbf{d}^\ell\in \cap\Z^I\text{ with }0<\beta_\ell<1\text{ for }\ell=0,\ldots,m.
\end{equation}
  
By Claim \ref{cla:maximal}, we know that $\mathcal{D}$ spans $\R^I$.  Since $\mathcal{D}$ is also  unimodular, by Fact \ref{fac:unimodSet} there exist $\mathbf{d}^{m+1},\ldots,\mathbf{d}^n\in\mathcal{D}$ for some $n\geq m$ such that $\mathbf{d}^1,\ldots,\mathbf{d}^n$ are linearly independent and
\[
\mathbf{z}=\sum_{\ell=1}^n\gamma_\ell \mathbf{d}^\ell \text{ with }\gamma_\ell \in\Z \text{ for all }\ell=1,\ldots,n.
\]
Moreover, by replacing $\mathbf{d}^{m+1},\ldots,\mathbf{d}^n$ with their negations if necessary, we can assume that $\gamma_{m+1},\ldots,\gamma_n \geq 0$.  

Let $k \in J \ssm \{j\}$ be arbitrary. Let $X^k_I=Y^k_I+Z^k_I$, where
\begin{align*}
Y^k_I	&=\left\{\rgivend\sum_{\ell=1}^m\alpha_\ell\mathbf{d}^\ell\rgiv -|\gamma_\ell|\leq\alpha_\ell\leq |\gamma_\ell|+1 \text{ for }\ell=1,\ldots,m\right\}\cap\Z^I\\
Z^k_I	&=\left\{\rgivend\sum_{\ell=m+1}^n\alpha_\ell\mathbf{d}^\ell\rgiv 0\leq\alpha_\ell\leq \gamma_\ell \text{ for }\ell=m+1,\ldots,n\right\}\cap\Z^I.
\end{align*}
Observe that $\mathbf{z}\in X^k_I$.  Moreover, $\conv(X^k_I)\cap\Z^I=X^k_I$, as we may see by writing $X^k_I\coloneq\{\sum_{\ell=1}^n\alpha_\ell \mathbf{d}^\ell|c_\ell\leq\alpha_\ell\leq d_l \text{ for }\ell=1,\ldots,n\}\cap\Z^I$ for suitably chosen $c_\ell$ and $d_\ell$.

Choose $\mathbf{s}_I$ such that $\mathbf{s}_I\cdot\mathbf{d}^\ell=0$ for $\ell=1,\ldots,m$ and $\mathbf{s}_I\cdot\mathbf{d}^\ell<0$ for $\ell=m+1,\ldots,n$.  (Such an $\mathbf{s}_I$  exists because $\mathbf{d}^1,\ldots,\mathbf{d}^n$ are linearly independent.)  Set $\mathbf{t}_I\coloneq\p+\mathbf{s}_I$ and set $V^k\coloneq V^{k,\mathbf{t}_I}$.  Then $V^k$ is concave by Lemma \ref{lem:linConv}.  By Equation (\ref{eqn:faces}) and the definition of $X^k_I$, we deduce that
\[
D^k(\p)=\argmax_{\mathbf{x}_I\in X^k_I}\mathbf{s}_I\cdot\mathbf{x}_I
=Y^k_I.
\]
Moreover, the edges of $X^k_I$ are parallel to $\mathbf{d}^1,\ldots,\mathbf{d}^n$ and so by Corollary \ref{cor:edges}, the valuation $V^k$ is of demand type $\mathcal{D}$.  

For agents $j' \in J \ssm \{j,k\},$ let $\Feas{j'} = \{\zero\}$, let $\valFn{j'}$ be arbitrary, and let $\bundowag{j'} = \mathbf{0}$.
Let the total endowment $\mathbf{y}_I$ be $\mathbf{x}'_I + \mathbf{z}$.  Set $\mathbf{w}^j_I \coloneq \mathbf{x}'_I \in X^j_I$ and $\mathbf{w}^k_I = \mathbf{z} \in X^k_I,$ so $(\mathbf{w}^{j'}_I)_{j' \in J}$ is \andowalloc.

Now see that
\begin{equation}\label{eqn:demSet}
\sum_{j'\in J}D^j(\p)\subseteq\{\bunpr+\alpha\norm\ | \ \alpha=0,\ldots,r\}+Y^k_I
\end{equation}
while, since $\bunpr+r\norm\in D^j(\p)$, we have the equality
\[
\conv\left(\sum_{j'\in J}D^j(\p)\right)=\{\bunpr+\alpha\norm\ | \ 0\leq\alpha\leq r\}+\conv(Y^k_I)
\]
Recalling Equation (\ref{eqn:zinP}), we conclude that $\tot = \bunpr+\mathbf{z}\in \conv\left(\sum_{j'\in J}D^j(\p)\right)$, so $\p$ is a pseudo-equilibrium price vector.  But, since $0<\beta_0<1$ in Equation (\ref{eqn:zinP}) and since the set $\{\mathbf{d}^1,\ldots,\mathbf{d}^m,\mathbf{g}\}$ is linearly independent, we conclude from Equation (\ref{eqn:demSet}) that $\tot = \bunpr+\mathbf{z}\notin \sum_{j'\in J}D^j(\p)$, so there is no \ce\ at $\p$.
%
Therefore, by the contrapositive of Fact~\ref{fac:pseudoEquil}, no \ce\ can exist.
\end{casework}
As the cases exhaust all possibilities, we have proven the fact.

\end{document}